\documentclass[sigconf]{acmart}
\usepackage{color}
\usepackage{soul}
\usepackage{graphicx}
\usepackage[ruled,vlined]{algorithm2e}
\usepackage[noend]{algorithmic}
\usepackage{chngcntr}
\usepackage{epstopdf}
\usepackage{caption}
\usepackage[labelformat=simple]{subcaption}

\makeatletter
\newcounter{module}
\newenvironment{module}[1][htb]
  {
   \let\c@algocf\c@module
   \begin{algorithm}[#1]%
  }{\end{algorithm}}
\makeatother

\newcommand{\D}{\mathrm{\textbf{D}}}
\newcommand{\tet}{\vec{\theta}}
\newcommand{\A}{\mathrm{\textbf{A}}}
\newcommand{\B}{\mathrm{\textbf{B}}}

\newcommand{\sgn}{\mathrm{sgn}}

\newcommand{\W}{\mathrm{\textbf{W}}}
\newcommand{\Q}{\mathrm{\textbf{Q}}}

\newcommand{\diag}{\mathrm{diag}}

\newcommand{\Y}{\mathrm{\textbf{Y}}}

\title{Protecting the Grid against IoT Botnets of High-Wattage Devices}
\author{Saleh Soltan, Prateek Mittal, H. Vincent Poor}
\affiliation{%
  \institution{Department of Electrical Engineering, Princeton University}}
\email{{ssoltan,pmittal,poor}@princeton.edu}

\setcopyright{none}
\renewcommand\footnotetextcopyrightpermission[1]{} 
\begin{document}
\setlength{\textfloatsep}{2 pt}

\begin{abstract}
We provide methods to prevent line failures in the power grid caused by a newly revealed \underline{MA}nipulation of \underline{D}emand (MAD) attacks via an IoT botnet of high-wattage devices. In particular, we develop two algorithms named \underline{S}ecuring \underline{A}dditional margin \underline{F}or generators in \underline{E}conomic dispatch (SAFE) Algorithm and \underline{I}teratively \underline{M}ini\underline{M}ize and bo\underline{UN}d \underline{E}conomic dispatch (IMMUNE) Algorithm for finding robust operating points for generators during the economic dispatch such that no lines are overloaded after automatic primary control response to any MAD attacks. In situations that the operating cost of the grid in a robust state is costly (or no robust operating points exist), we provide efficient methods to verify--in advance--if possible line overloads can be cleared during the secondary control after any MAD attacks. We then define the $\alpha D$-robustness notion for the grids indicating that any line failures can be cleared during the secondary control if an adversary can increase/decrease the demands by $\alpha$ fraction. We demonstrate that practical upper and lower bounds on the maximum $\alpha$ for which the grid is $\alpha D$-robust can be found efficiently in polynomial time. Finally, we evaluate the performance of the developed algorithms and methods on realistic power grid test cases. Our work provides the first methods for protecting the grid against potential line failures caused by MAD attacks.


\end{abstract}
\maketitle

\section{Introduction}\label{sec:introduction}

Power grid, as one of the most essential infrastructure networks, has been repeatedly evinced in the past couple of years to be vulnerable to cyber attacks. The most infamous example of these attacks was on Ukrainian grid that affected about 225,000 people in December 2015~\cite{UkraineBlackout}. However, smaller scale attacks on reginal power grids have been shown in a recent report to be more common and pervasive~\cite{cyberattackUS}. As indicated in the report, \emph{``Hackers are developing a penchant for attacks on energy infrastructure because of the impact the sector has on people's lives"}~\cite{cyberattackUS}.

Because of this ever-growing threat, there has been a significant effort by researchers in recent years to protect the grid against cyber attacks. These efforts has been mainly focused on potential attacks that directly affect different components of power grids' Supervisory Control And Data Acquisition (SCADA) systems. Many system operators prefer to completely disconnect their SCADA systems from the Internet in the hope that their systems remain unreachable to hackers. 

Despite these efforts, the \emph{power demand} side of the grid operation which is not controlled by SCADA, has been justifiably neglected in security assessments due to their predictable nature. However, as recently revealed by Soltan et al.~\cite{soltan2018blackIoT} and Dabrowski et al.~\cite{dabrowski2017grid}, the universality and growth in the number of high-wattage Internet of Things (IoT) devices, such as air conditioners and water heaters, have provided a unique way for adversaries to \emph{disrupt the normal operation of power grid, without any access to the SCADA system.} In particular, an adversary with access to sufficiently many of such high-wattage devices (i.e, a \emph{botnet}), can abruptly increase or decrease the total demand in the system by synchronously turning these devices on or off, respectively. We call these attacks \underline{MA}nipulation of \underline{D}emand (MAD) attacks (see Fig.~\ref{fig:idea}).

\begin{figure}[t]
\centering
\includegraphics[scale=0.35]{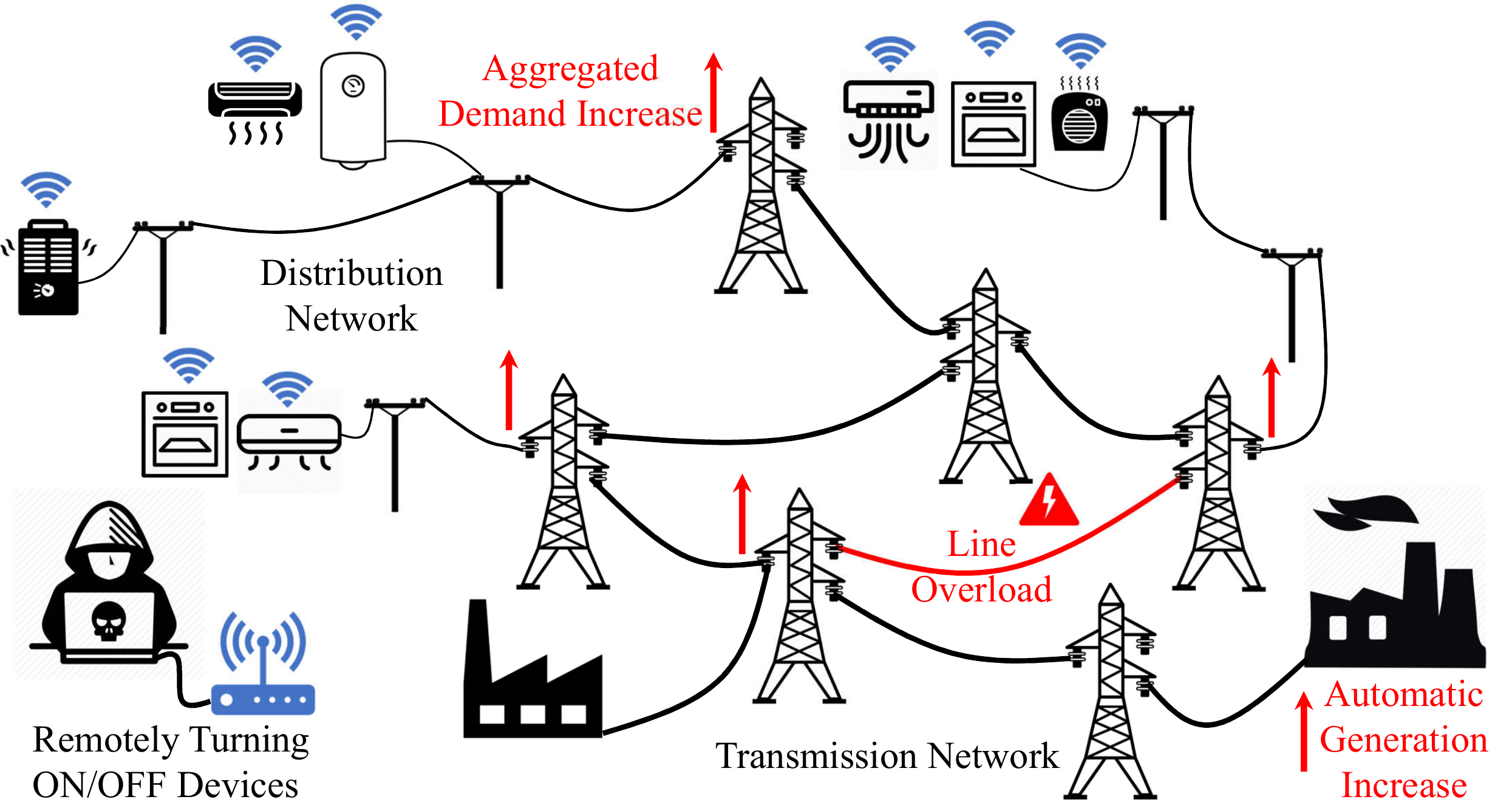}
\caption{The MAD attack. An adversary with access to an IoT botnet of high-wattage devices can remotely and synchronously switch on/off these devices in order to change power flows on the lines in power grid transmission network and cause line overloads and failures.} 
\label{fig:idea}
\end{figure}

An abrupt increase/decrease in the total demand, results in abrupt drop/rise in the system's frequency. If this drop/rise is significant, 
generators will be automatically disconnected from the grid and a large scale blackout occurs within seconds~\cite{soltan2018blackIoT,dabrowski2017grid}. 
If the drop/rise in the frequency is not significant, the extra demand/generation can automatically be compensated by generators' primary controllers, and the frequency of the system will be stabilized. As a result of this automatic change in generation--and demand by the adversary--the power flows in the transmission network change based on power flow equations. Since the power flows are not controlled by the grid operator at this stage, this change in the power flows may result in line overloads and consequent line-trippings. These initial line failures can initiate a cascading line failure and result in a large scale black out in the grid~\cite{soltan2018blackIoT}. 

The grid operator can protect the grid against initial drop/rise in the system's frequency caused by a MAD attack by ensuring that the system have enough \emph{inertia} (mostly through rotating generators) and there is enough available \emph{spinning reserve} (i.e., generators have enough extra generation capacity)~\cite{soltan2018blackIoT}. However, protecting the grid against possible line overloads and failures after a MAD attack, which is the main focus of this paper, is more analytically and computationally challenging. Such defences require the grid operator to analyze all possible MAD attacks and their consequences on the power flows and select operating points for the generators such that no lines are overloaded after any MAD attacks. 

We first focus on finding operating points (namely \emph{robust operating points}) with the minimum cost for the generators such that no lines are overloaded after the automatic primary response of the generators to any MAD attacks.  Since changes in power flows after a MAD attack directly depend on generators' operating points, finding the optimal operating points for the generators requires solving a nonconvex and nonlinear optimization problem which is hard in general. Despite this hardness, we develop two algorithms named \underline{S}ecuring \underline{A}dditional margin \underline{F}or generators in \underline{E}conomic dispatch (SAFE) Algorithm and \underline{I}teratively \underline{M}ini\underline{M}ize and bo\underline{UN}d \underline{E}conomic dispatch (IMMUNE) Algorithm for finding suboptimal yet robust operating points for the generators efficiently. The SAFE Algorithm provides robust operating points for the generators by solving a single Linear Program (LP). The IMMUNE Algorithm on the other hand requires a few iterations until it converges, but it provides robust operating points with lower costs than the ones obtained by the SAFE Algorithm.

In situations that the operating cost of the grid in a robust state is costly (or no robust operating points exist due to lack of enough resources), the grid operator may decide to allow temporary line overloads--by increasing thresholds on circuit breakers--in the case of a MAD attack, and clear the overloads during the \emph{secondary control}. During the secondary control, which comes right after the automatic primary control, the grid operator can directly change generators' operating points in order to bring back the system's frequency to its nominal value and clear any line overloads. To make sure that line overloads can be cleared during the secondary control, the grid operator needs to verify in advance whether for any potential MAD attack, there exist operating points for the generators satisfying demands such that no lines are overloaded (namely, the grid is \emph{secondary controllable}). However, due to the extent of the attack space, checking all possible attack scenarios is numerically impossible. Hence, we develop several predetermined control polices that can be used to verify the secondary controllability of the grid in most scenarios with no false positives.

Based on the secondary controllability notion, we then evaluate the robustness of a grid against MAD attacks. A magnitude of a MAD attack is defined based on the size of an adversary's botnet. If these bots are distributed proportionally to the demand sizes across a region, then the magnitude of an attack can be denoted by $\alpha$ fraction of demand that the adversary can increase or decrease at each location. We call a grid that is secondary controllable against any MAD attacks that change at most $\alpha$ fraction of demands, $\alpha D$-robust. In general, finding maximum $\alpha$ such that a given grid is $\alpha D$-robust, is hard. However, based on developed predetermined control policies, we provide efficient methods to compute practical upper and lower bounds for the maximum $\alpha$.

Finally, we numerically evaluate the performance of the developed algorithms and controllers. For example, in New England 39-bus system, we show that the SAFE and IMMUNE Algorithms find operating points for the generators with at most 6 and 2 percent increase in the total operating cost such that the grid is robust against MAD attacks of magnitude $\alpha =0.08$. We also evaluate the performance of the developed methods for approximating the maximum $\alpha$ that grid is $\alpha D$-robust and show that for example in New England 39-bus system, the provided lower and upper bounds are tight and are equal to the maximum $\alpha^{\max}=0.0962$.

\emph{To the best of knowledge, our work is the first to study the effects of potential MAD attacks via high-wattage IoT devices on the power flows in the grid and provide efficient preventive algorithms to avoid line failures after primary control response, and also efficient methods to verify if the line overloads can be cleared during the secondary control. These algorithms and methods can be adopted by grid operators to protect their systems against MAD attacks now and in the near future.}

The rest of this paper is organized as follows: Section~\ref{sec:related} provides related work and Section~\ref{sec:model} presents a brief introduction to power system's operation and control. In Section~\ref{sec:consequences}, we introduce the MAD attacks and provide their basic properties. In Section~\ref{sec:primary}, we present the SAFE and IMMUNE algorithms and in Section~\ref{sec:secondary}, we provide efficient methods for verifying secondary controllability of a grid. Section~\ref{sec:contingency} provides methods to evaluate robustness of grids against MAD attacks and Section~\ref{sec:numerical} presents numerical results. Finally, Section~\ref{sec:conclusions} provides concluding remarks and future directions. Due to space constraints, some proofs are presented in the Appendix.
\section{Related Work}\label{sec:related}

Power systems' vulnerability to failures and attacks has been widely studied in the past few years~\cite{dobson2015cascading,soltan2015analyz,hale2016ACDC,bienstock2016electrical,carreras2002critical}. In a recent work~\cite{garcia2017hey}, Garcia et al.\ introduced Harvey malware that affects power grid control systems and can execute malicious commands.
Theoretical methods for detecting cyber attacks on power grids and recovering information after such attacks have also been developed~\cite{liu2011false,dan2010stealth,li2015quickest,kim2015subspace,SYZ2015,bienstock2017computing,li2016bilevel,deng2017ccpa,zhang2016physical}. However, most of the previous work have focused on attacks that directly target the power grid's physical infrastructure or its control system.


Load altering attacks on smart meters and large cloud servers has been first introduced by Mohsenian et al.~\cite{mohsenian2011distributed}. Their work was mostly focused on the cost of protecting the grid against such attacks at loads. 
Amini et al.~\cite{amini2016dynamic} have also recently studied the effects of load altering attacks on the dynamics of the system and ways to use the system's frequency as feed-back to improve an attack. In three very recent papers, Dvorkin and Sang~\cite{dvorkin2017IoT}, Dabrowski et al.~\cite{dabrowski2017grid}, and Soltan et al.~\cite{soltan2018blackIoT} revealed the possibility of exploiting compromised IoT devices to manipulate the demands and to disrupt normal operation of the power grid. Dvorkin and Sang~\cite{dvorkin2017IoT} modeled their attack as an optimization problem for the adversary--with complete knowledge of the grid--to cause circuit breakers to trip in the distribution network. 
Dabrowski et al.~\cite{dabrowski2017grid} studied the effect of demand increases caused by remote activation of CPUs, GPUs, hard disks, screen brightness, and printers on the frequency of the European power grid. Soltan et al.~\cite{soltan2018blackIoT} analyzed the effects of sudden increase and decrease in the demand via an IoT botnet of high-wattage devices from different operational perspectives and demonstrated that besides frequency instability, such attacks can also result in wide spread cascading line failures. \emph{To the best of our knowledge, however, the work presented in this paper is the first to provide practical defences against possible line failures caused by attacks that target the demand side of power grids.}

\section{Model and Definitions}\label{sec:model}
In this section, we provide a brief introduction to power systems' operation and control. Our focus is on power transmission network.

Throughout this paper we use bold uppercase characters to denote matrices (e.g., $\textbf{A}$), italic uppercase characters to denote sets (e.g., $V$), and italic lowercase characters and overline arrow to denote column vectors (e.g., $\vec{\theta}$). For a matrix $\textbf{Q}$, $\textbf{Q}_i$ denotes its $i^{\text{th}}$ row, $q_{ij}$ denotes its $(i,j)^{\text{th}}$ entry, and $\textbf{Q}^T$ denotes its transpose. For a column vector $\vec{y}$, $\vec{y}^T$ denotes its transpose, and $\|\vec{y}\|_1:=\sum_{i=1}^n |y_i|$ is its $l_1$-norm. 
For a variable $x$, $\sgn(x)$ denotes its sign, and $\overline{x}$ and $\underline{x}$ denote its upper and lower limits, respectively. For a vector $\vec{y}$, for simplicity of notation, we drop the vector sign $~\vec{}~$ in denoting vectors of upper and lower limits on the entries of $\vec{y}$ as $\overline{y}$ and $\underline{y}$, respectively. Finally, $\vec{e}_1,\dots,\vec{e}_n$ denote the fundamental basis of $\mathbb{R}^n$ and $\vec{1}=\sum_{i=1}^n \vec{e}_i$ denotes the all ones vector.
\subsection{Power Flows}\label{subsec:powerflows}
Power flows are governed by a set of differential equations. In the steady-state, using \emph{phasors}, these differential equations can be reduced to a set of algebraic equations on complex numbers known as \emph{Alternating Current (AC)} power flow model. Due to nonlinearity of AC power flow equations and the computational complexity of solving these equations, in practice and in day-ahead power grid contingency analysis and planning, the linearized version of these equations known as \emph{Direct Current (DC)} power flow model is widely being used~\cite{wood2012power}. Hence, in this work, for simplicity of presentation, we also adapt the DC power flow model for our analysis. However, the main idea of the algorithms developed in this work can be extended to the AC power flow models as well.

We represent the power grid by a connected directed graph $G=(V,E)$ where $V=\{1,2,\dots,n\}$ and $E=\{e_1,\dots,e_m\}$ are the set of nodes and edges corresponding to the \emph{buses} and \emph{transmission lines}, respectively (the definition implies $|V|=n$ and $|E|=m$). Each edge $e$ is a set of two nodes $e=(i,j)$. (Direction of the edges are arbitrary.)
$\vec{p_d}\geq 0$ and $\vec{p_g}\geq 0$ denote the vector of power demand and supply values, respectively. Accordingly, $\vec{p}=\vec{p_g}-\vec{p_d}$ denotes the vector of total supply and demand values. Since sum of supply should be equal to sum of demand,
\begin{equation}\label{eq:p}
\vec{1}^T\vec{p}=0,
\end{equation}
in which $\vec{1}$ is an all ones vector.
In the DC model, lines are also assumed to be \emph{purely reactive}, implying that each edge $e=(i,j) \in E$ is characterized by its \emph{reactance} $x_e=x_{ij}>0$.

Given the power supply/demand vector $\vec{p}\in \mathbb{R}^{n\times1}$ and the reactance values, the vector of power flows on the lines $\vec{f}\in \mathbb{R}^{m\times1}$ can be computed by solving following linear equations:
\begin{eqnarray}
\label{eqn:flow1}&&\A\tet=\vec{p},\\
\label{eqn:flow2}&&\Y\D^T\tet = \vec{f},
\end{eqnarray}
where $\tet\in\mathbb{R}^{n\times 1}$ is the vector of voltage phase angles at nodes, 
$\D\in\{-1,0,1\}^{n\times m}$ is the \emph{incidence matrix} of $G$ defined as,
\begin{equation*}
d_{ik}=
\begin{cases}
0&\text{if}~e_k~\text{is not incident to node}~i,\\
1&\text{if}~e_k~\text{is coming out of node}~i,\\
-1&\text{if}~e_k~\text{is going into node}~i,
\end{cases}
\end{equation*}
$\Y:=\diag([1/x_{e_1},1/x_{e_2},\dots, 1/x_{e_m}])$ is a diagonal matrix with diagonal entries equal to the inverse of the reactance values, and  $\A=\D\Y\D^T$ is the \textit{admittance matrix} of $G$.\footnote{The admittance matrix $\textbf{A}$ is also known as the \emph{weighted Laplacian matrix} of the graph~\cite{bapat2010graphs} in graph theory.}

Since $\A$ is not a full-rank matrix, we follow~\cite{soltan2015analyz} and use the \emph{pseudo-inverse} of $\A$, denoted by $\A^+$ to solve (\ref{eqn:flow1}) as $\tet=\A^+\vec{p}$. Once $\vec{\theta}$ is computed, $\vec{f}$ can be computed from~(\ref{eqn:flow2}) as $\vec{f} = \Y\D^T\A^+\vec{p}$. For convenience of notation, we define $\B:= \Y\D^T\A^+$. Hence, $\vec{f}=\B\vec{p}$.

\subsection{Power Grid Operation}\label{subsec:opf}

Stable operation of the power grid  relies on the persistent balance between the power supply and  demand. 
In order to keep the balance between the power supply and the demand, power system operators use weather data as well as historical power consumption data to predict the power demand on a daily and hourly basis~\cite{federal2012energy}. This allows the system operators to plan in advance and only deploy enough generators to meet the demand in the hours ahead without overloading any power lines. This planning ahead consists of two parts: \emph{unit commitment} and \emph{economic dispatch}.

In unit commitment which is mainly performed daily, the grid operator selects a set of generators to \emph{commit} their availability during the day-ahead operation of the grid. But the actual operating points of the generators (i.e., generation outputs) are determined by the operator during the day and in the process known as \emph{economic dispatch}. The main goal of operator during economic dispatch is to ensure reliable operation of the grid with minimum power generation cost. When feasibility of the power flows is also considered during economic dispatch, the process is also known as \emph{Optimal Power Flow (OPF)} problem. Since in practice feasibility of power flows is always being considered, these two terms can be used interchangeably most of the times.

In this work, we mainly focus on ensuring robustness of the grid during the economic dispatch. Extending our methods to unit commitment process is beyond the scope of this paper and is part of the future work. Hence, here we assume that set of available generators are given. The main challenge is to obtain a favorable operating point for these generators.


\subsubsection{Optimal Power Flow}
In the OPF problem, given the vector of predicted demand values $\vec{p_d}$, the grid operator needs to find the operating point vector $\vec{p_g}$ for the generators such that supply matches the demand (i.e., $\vec{1}^T(\vec{p_g}-\vec{p_d})=0$), the operating and physical constraints are satisfied, and the operating cost of the generators are minimized.

In particular, each line $f_{ij}$ has a thermal power flow limit $\overline{f_{ij}}$ limiting the amount of power that a line can \emph{safely} carries. If the power flow on a line goes above this limit (i.e., \emph{overloads}), in most of the cases, it will be tripped by a circuit breaker in order to keep the line from breaking due to overheating. Hence, during the normal operation of the grid
\begin{equation}\label{eq:f}
|f_{ij}|\leq \overline{f_{ij}}, \quad \forall (i,j)\in E.
\end{equation}
The amount of power that each generator $p_{gi}$ is generating is also limited by a maximum ($\overline{p_{gi}}$) and a minimum ($\underline{p_{gi}}$) value. If there are no generators at node $i$, then $\overline{p_{gi}}=\underline{p_{gi}}=0$. Hence,
\begin{equation}\label{eq:pg}
\underline{p_g}\leq \vec{p_g}\leq \overline{p_g}.
\end{equation}

The generation cost at each generator is a given by a cost function $c_i(x)$ in $\$/hr$. Given these cost functions, the OPF problem can be formulated as follows:
\begin{align}
\label{eq:dcopf}&\min_{\tet,\vec{f},\vec{p_g}}&&\sum_{l=1}^n c_l(p_{gi}),\\
&\qquad\text{s.t.}& &(\ref{eq:p}),(\ref{eqn:flow1}),(\ref{eqn:flow2}), (\ref{eq:f}), (\ref{eq:pg}),\nonumber\\
&&&\vec{p}=\vec{p_g}-\vec{p_d}.\nonumber
\end{align}
Several methods for optimally solving (\ref{eq:dcopf}) depending on the cost functions exist in literature~\cite{wood2012power}. 
Hence, we assume that (\ref{eq:dcopf}) can be solved either optimally or approximately using existing methods. 
%
\subsection{Frequency control}\label{subsec:frequency}
In power systems, the rotating speed of generators correspond to the frequency. When demand becomes greater than supply, the rotating speeds of turbine generators' rotors decelerate, and the kinetic energy  of the rotors are released into the system in response to the extra demand. Correspondingly, this causes a drop in the system's frequency. This behavior of turbine generators corresponds to Newton's first law of motion and is calculated by the \emph{inertia} of the generators. Similarly, the supply being greater than the demand results in acceleration of the generators' rotors and a rise in the system's frequency. 

This decrease/increase in the frequency of the system cannot be tolerated for a long time since frequencies lower than their nominal value severely damage the generators. If the frequency goes above or below a threshold value, protection relays turn off or disconnect the generators completely. 
Hence, in case of a demand increase, within seconds of the first signs of decrease in the frequency, the \emph{primary controllers} at generators activate and increase the mechanical input to the generators which increase the speed of the generator's rotor and correspondingly the generator's output and frequency of the system~\cite{entsoe2004continental}. The rate of decrease/increase in frequency of the system, before activation of the primary controllers, directly depends on the total \emph{inertia} of the system. Systems with higher number of rotating generators have higher inertia and therefore are more robust against sudden demand changes or generation losses.

The rate of increase in the output generation of generator $i$ during the primary control is determined by its \emph{governor droop characteristic} denoted by $R_i$~\cite[Chapter 9]{machowski1997power}. 
Hence, after a change in the total demand by $S_{\Delta p_d}$, primary controller of each generator $i$ increases its output with rate $R_i$ until the total generation is equal to the demand again. In particular, if none of the generators reach their generation limit, each generator $i$ will increase its generation by $S_{\Delta p_d}/R_i$. The amount of power that generators can provide during the primary control is called the \emph{spinning reserve} of the generators.

\begin{figure}[t]
\centering
\includegraphics[scale=0.28]{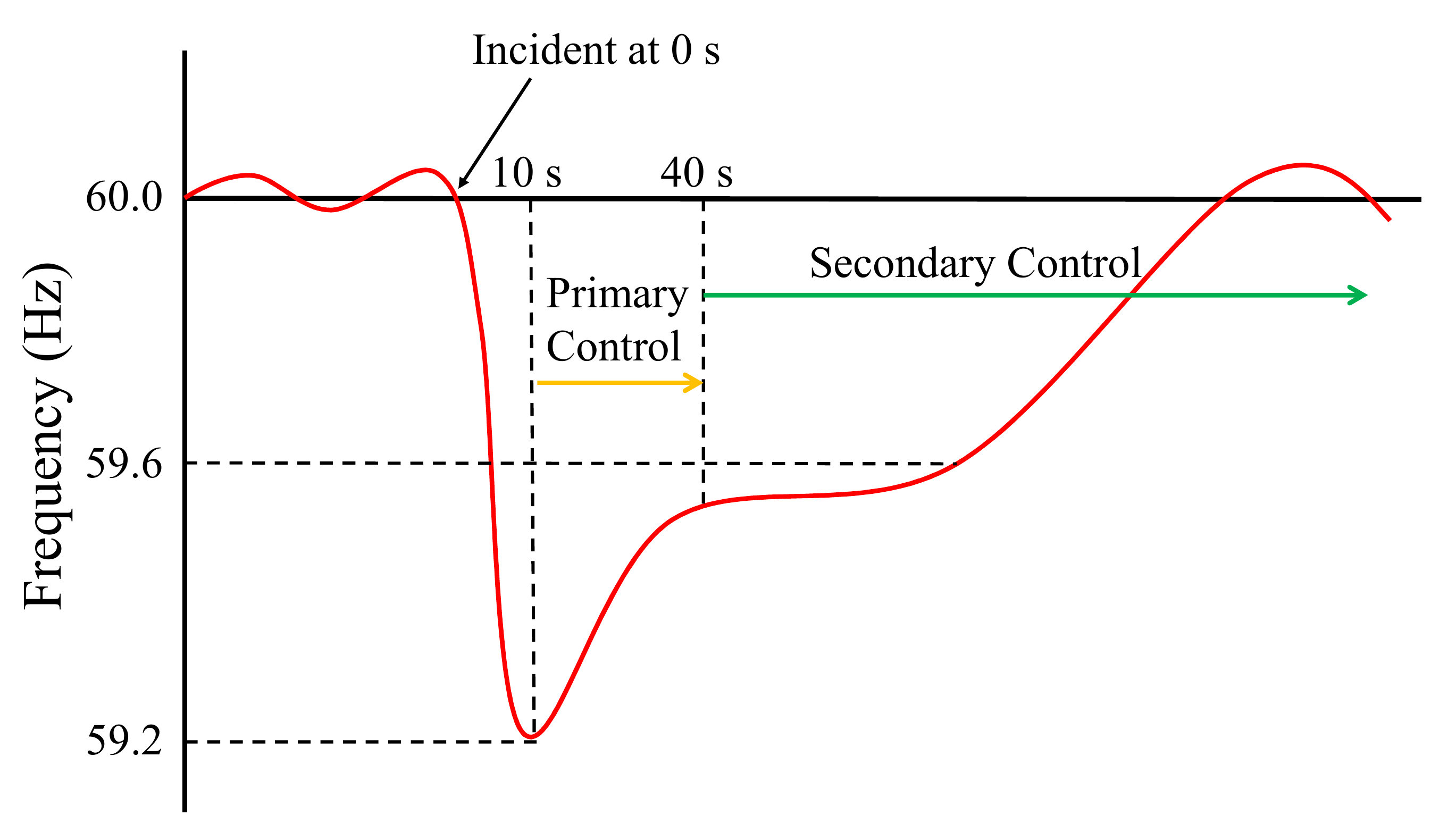}
\caption{A sample frequency response of the power grid to a sudden increase in the demand (or loss of generation).}
\label{fig:frequency}
\end{figure}



Despite stability of the system's frequency after the primary controllers' response, it may not return to its nominal value (since generators generating more than their generating set points). 
Hence, the \emph{secondary controller} starts within minutes to restore the system's frequency. The secondary controller modifies the power set points and deploys available extra generators and controllable demands to restore the nominal frequency and permanently stabilizes the system.\footnote{Part of these controls can be done during the \emph{tertiary control}. However, for simplicity and without loss of generality we refer to them as the secondary control.} 
Fig.~\ref{fig:frequency} presents an example of the way frequency of the system changes after a sudden increase in the demand (or loss of generation) at time 0.

\section{MAD Attacks}\label{sec:consequences}

We assume that an adversary has already gained access to an IoT botnet of many high-wattage smart appliances within a city, a country, or a continent. 
Such access can potentially allow the adversary to increase or decrease the demand at different locations \emph{remotely and synchronously} at a certain time. We call the attacks under this threat model the 
 \underline{MA}nipulation of the \underline{D}emand (MAD) attacks.

The adversary's power to manipulate the demand can be modeled by the maximum and minimum demand changes that it can make at each node. In particular, we assume the demand changes at node $l$ by an adversary are bounded by $-\overline{\Delta p_{dl}}\leq \Delta p_{dl}\leq \overline{\Delta p_{dl}}$. 
Notice that from defensive point of view, there are no differences between an adversary with the total knowledge of the system (a.k.a \emph{white-box} attacks) and an adversary with no knowledge of the system (a.k.a \emph{black-box} attacks), since the operator needs to make sure that the grid is robust against \emph{any possible attacks}.

The initial effect of a MAD attack, as described in Section~\ref{subsec:frequency} is on the frequency of the system. However, the system operator can make the system robust against frequency disturbances caused by MAD attacks by ensuring that enough generators with inertia and spinning reserve are committed to operate during the unit commitment process~\cite{soltan2018blackIoT}. The minimum required inertia and spinning reserve should be computed based on the potential attack size and the properties of the grid. Devices that provide virtual inertia such as batteries, super-capacitors, and flywheels can also be integrated into the system to increase the total inertia~\cite{hebner2002flywheel}. 

Hence, the main challenge in protecting the grid against initial effects of MAD attacks is in the hardware level. However, the effects of MAD attacks are not limited to frequency disturbances. Recall from Section~\ref{subsec:powerflows} that the power flows in power grids are determined uniquely given supply and demand values. Therefore, most of the time, the grid operator does not have any control over the power flows from generators to loads. Once an adversary causes a sudden increase in the loads all around the grid, assuming that the frequency drop is not significant, the extra demand is satisfied automatically by generators through their primary controllers as described in Section~\ref{subsec:frequency}. Since the power flows are not controlled by the grid operator at this stage, this change in supply and demand may result in line overloads and consequent line-trippings~\cite{soltan2018blackIoT}.

If the primary controllers' response results in line overloads, assuming that these overloads can barely be tolerated for a short period of time, these line overloads can be cleared during the secondary control. However, the system operator needs to ensure in advance that possible line overloads can indeed be cleared during the secondary control after any MAD attacks.

\emph{In this work, we focus on the effects of MAD attacks on the power flow changes on the lines which are more challenging from system planning perspective. Our objectives are: (i) to develop algorithms for finding efficient operating points for the generators during the economic dispatch such that no lines are overloaded after the primary control response to any potential MAD attacks, and (ii) to design methods to efficiently examine if line overloads after the primary control--if any--can be cleared during the secondary control.}

\section{Power Flows: Primary Control}\label{sec:primary}

In this section, we provide two algorithms for finding operating points for the generators during the economic dispatch process such that no lines are overloaded after automatic response of the primary controllers to any MAD attacks. We call such operating points, \emph{robust operating points}. 

\begin{figure*}[t]
\centering
\begin{subfigure}{0.245\textwidth}
\centering
\includegraphics[scale=0.4]{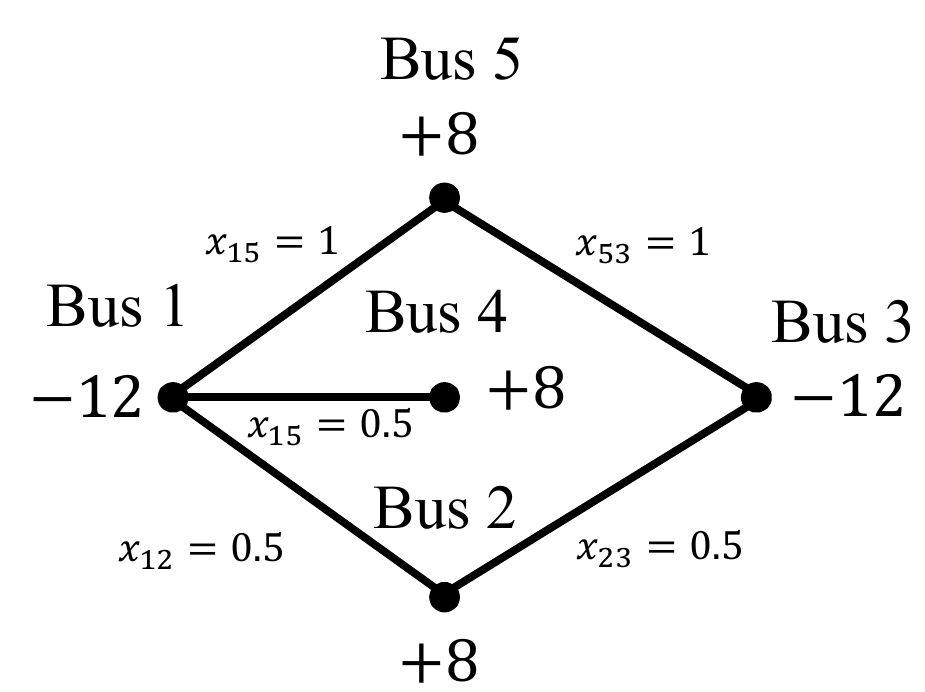}
\caption{}
\label{fig:exp1}
\end{subfigure}
\begin{subfigure}{0.245\textwidth}
\centering
\includegraphics[scale=0.4]{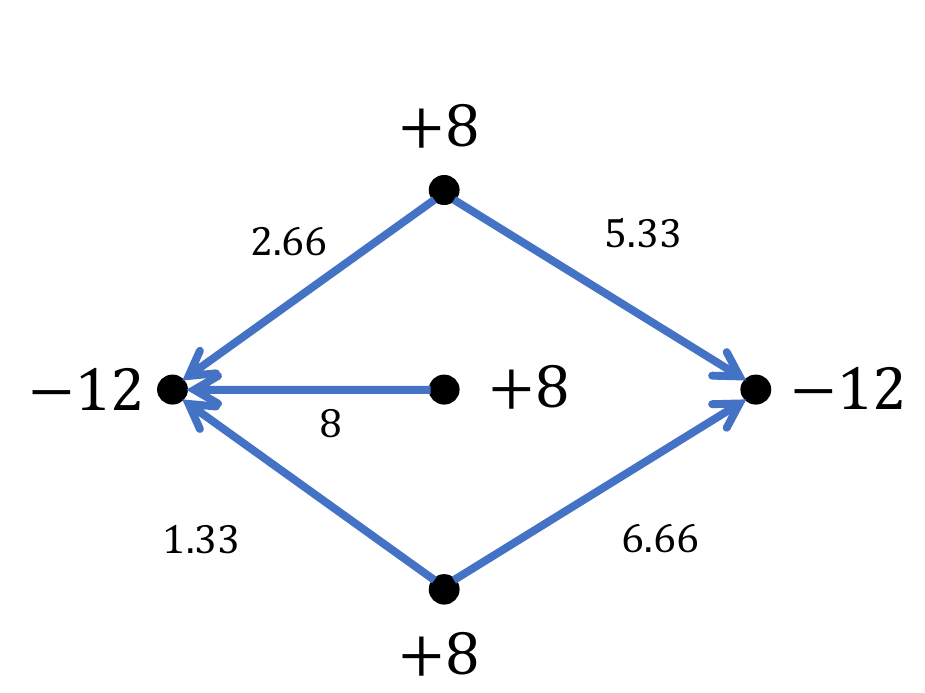}
\caption{}
\label{fig:exp2}
\end{subfigure}
\begin{subfigure}{0.245\textwidth}
\centering
\includegraphics[scale=0.4]{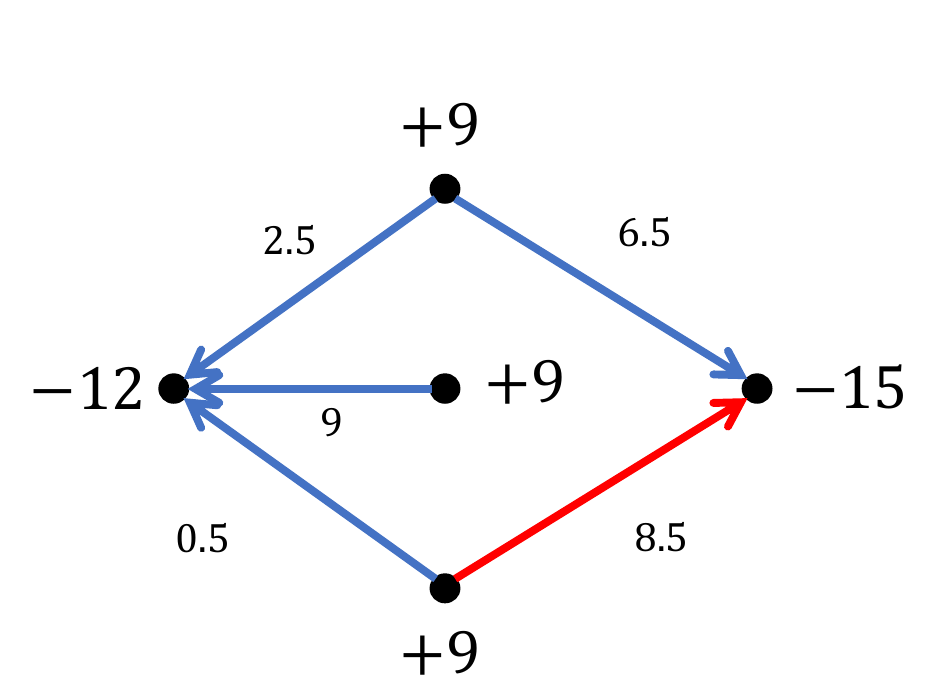}
\caption{}
\label{fig:exp3}
\end{subfigure}
\begin{subfigure}{0.245\textwidth}
\centering
\includegraphics[scale=0.4]{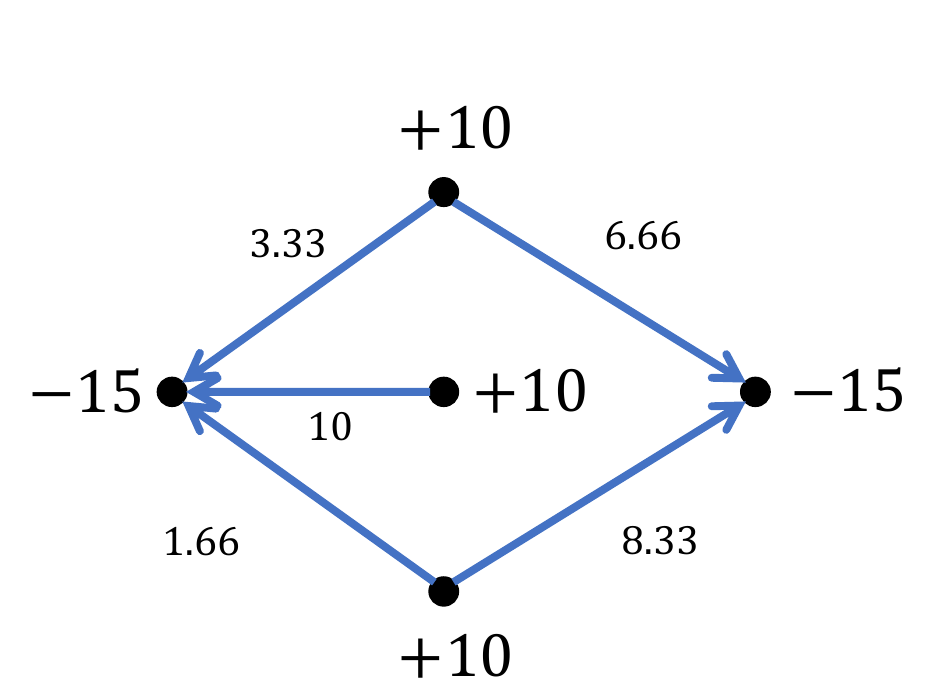}
\caption{}
\label{fig:exp4}
\end{subfigure}
\caption{An example demonstrating that increasing all demands may not necessarily result in the maximum flow on the lines. (a-b) Initial setting and power flows, (c) power flows if demand at bus 3 increases, and (d) power flows if demand at both buses 1 and 3 increases. All generators have the same droop characteristic and they all have enough spinning reserve.}
\label{fig:Example}
\end{figure*}
\begin{figure}[t]
\centering
\begin{subfigure}{0.23\textwidth}
\centering
\includegraphics[scale=0.4]{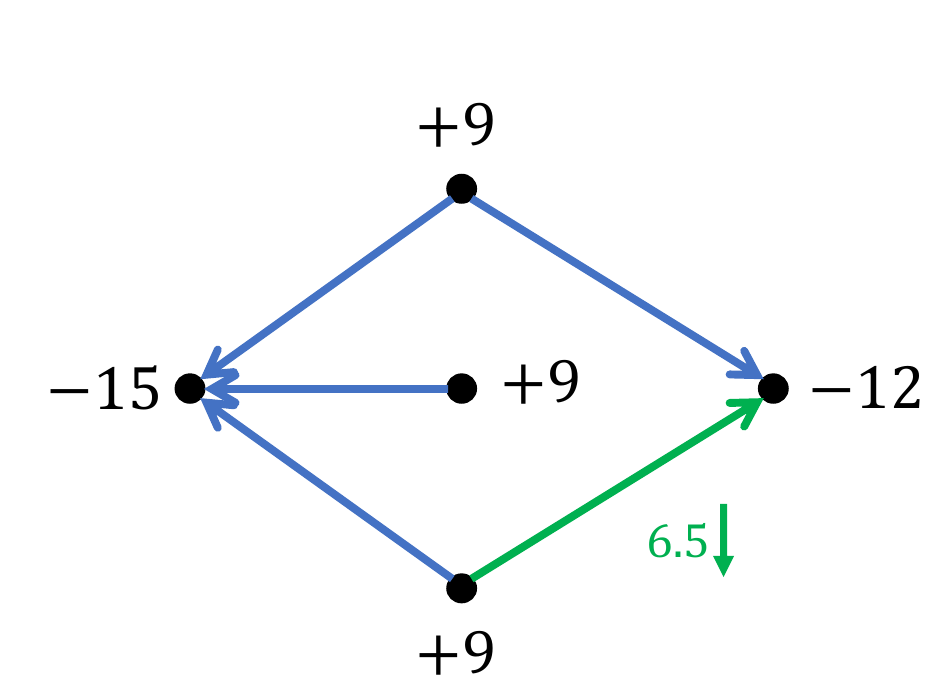}
\caption{}
\label{fig:exp21}
\end{subfigure}
\begin{subfigure}{0.23\textwidth}
\centering
\includegraphics[scale=0.4]{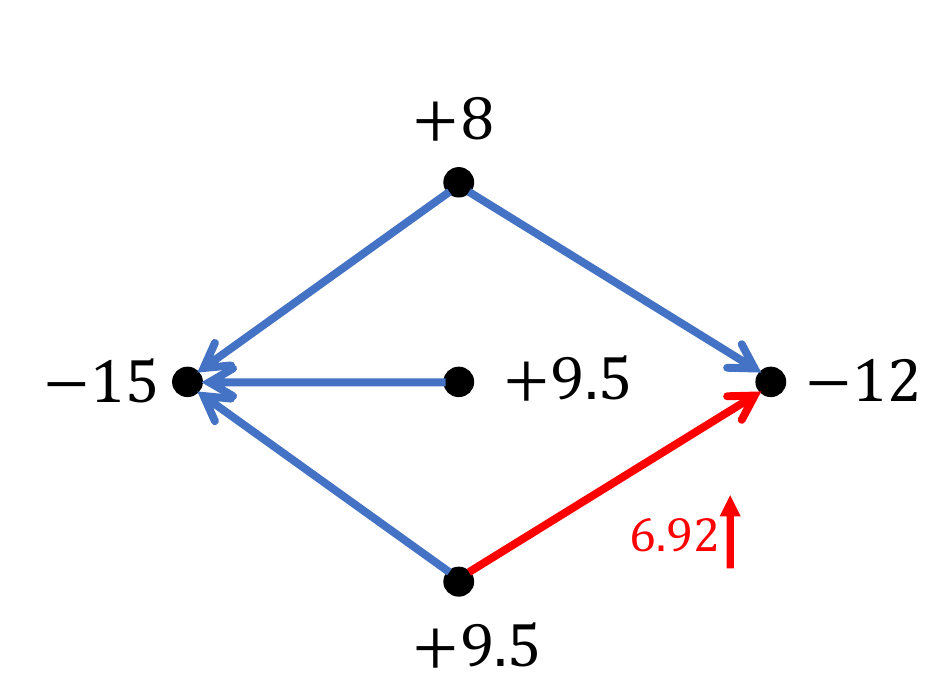}
\caption{}
\label{fig:exp22}
\end{subfigure}
\caption{Dependency of power flow changes on the location of the spinning reserves. (a) If all generators have spinning reserves, demand increase at bus 1 results in power flow decrease on line $\{2,3\}$. (b) If only generators 2 and 4 have spinning reserves then demand increase at bus 1 results power flow increase on line $\{2,3\}$.} 
\label{fig:Example2}
\end{figure}

\begin{figure}[t]
\centering
\begin{subfigure}{0.23\textwidth}
\centering
\includegraphics[scale=0.3]{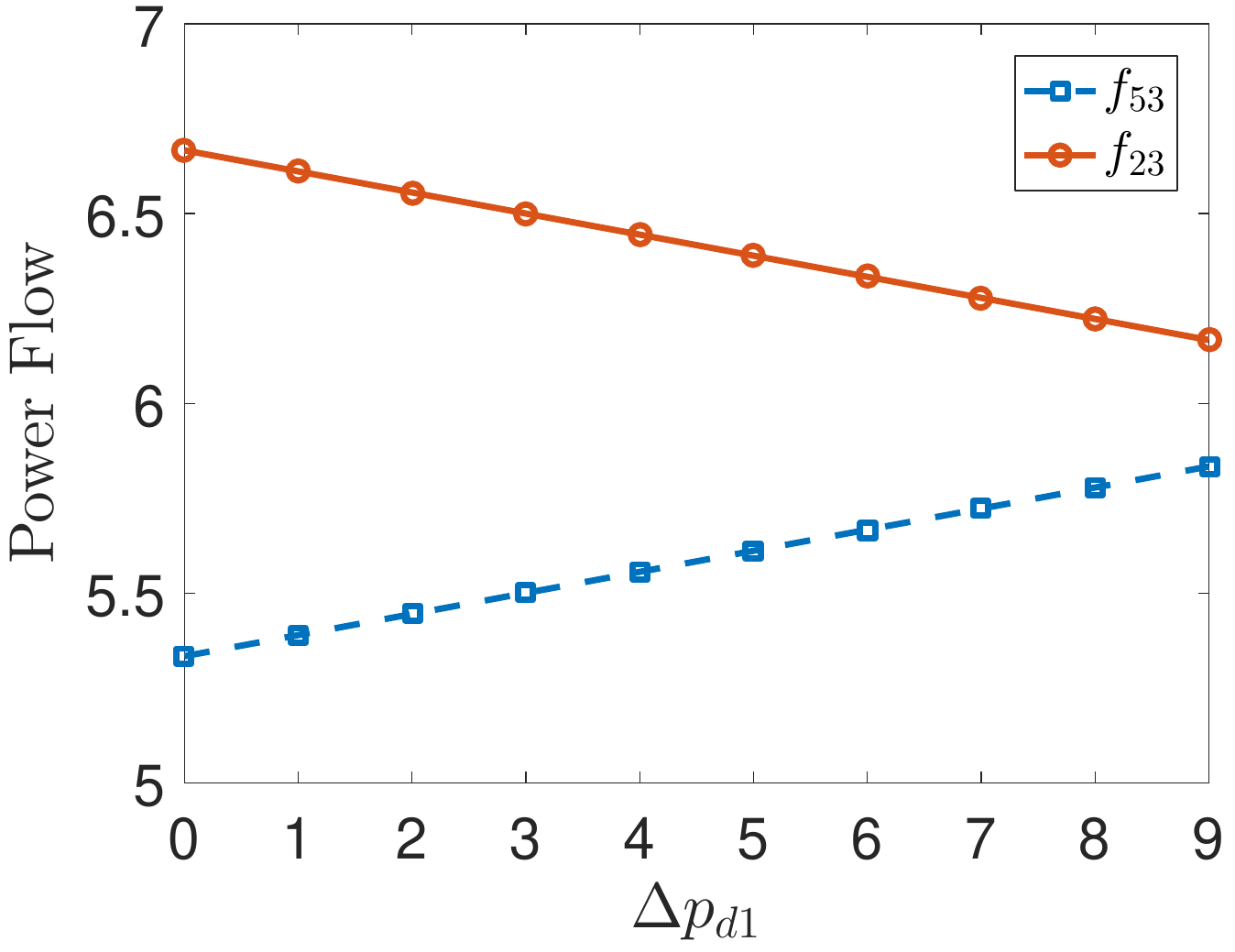}
\caption{}
\label{subfig:flow_spin_nolimit}
\end{subfigure}
\begin{subfigure}{0.23\textwidth}
\centering
\includegraphics[scale=0.3]{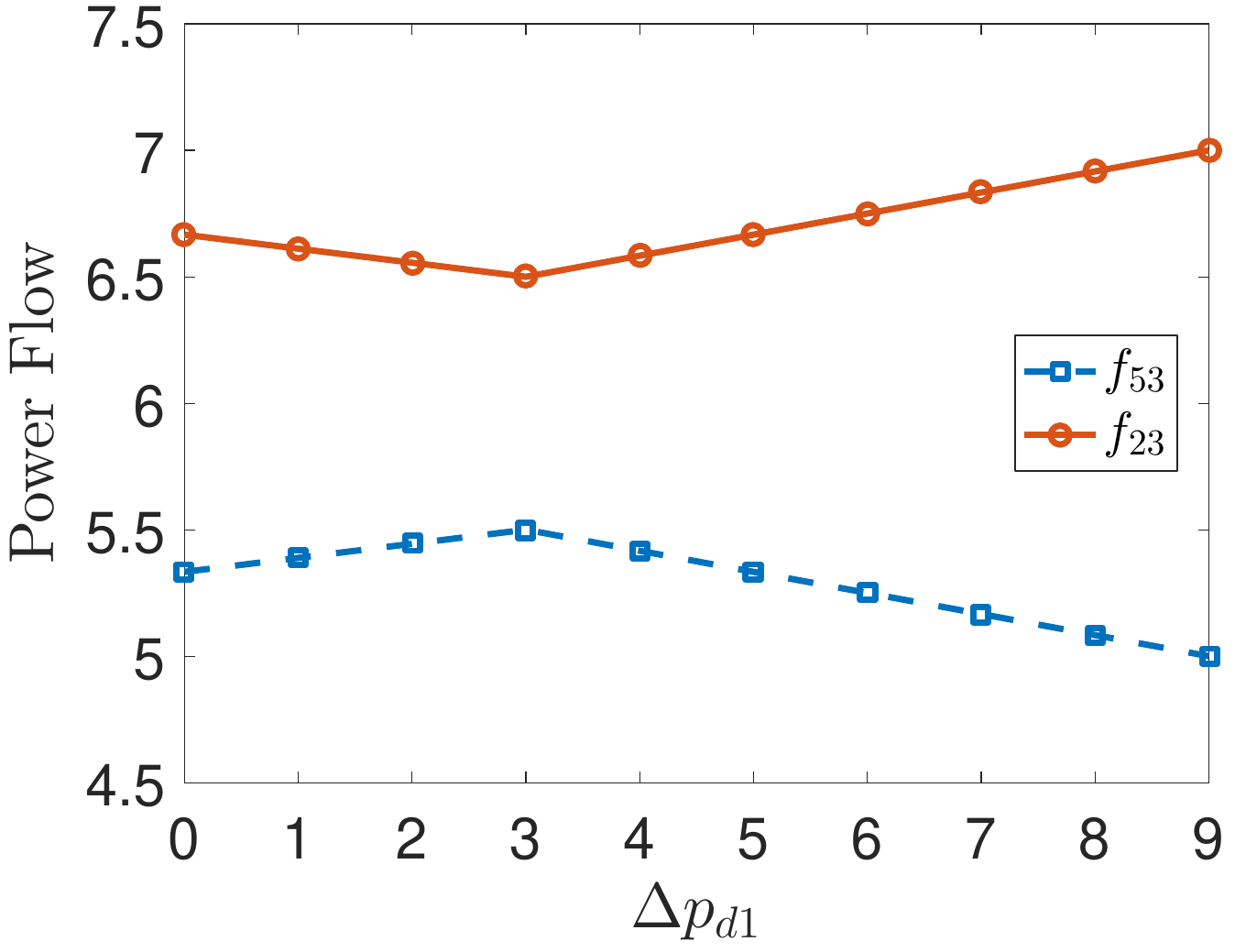}
\caption{}
\label{subfig:flow_spin_limit}
\end{subfigure}
\caption{Power flows on lines $\{5,3\}$ and $\{2,3\}$ in the grid shown in Fig.~\ref{fig:exp1} as demand at bus 1 increases. (a) If all the generators have enough spinning reserve, and (b) if generator 5 has only 1 unit of spinning reserve.}

\label{fig:Example_Flows}
\end{figure}





\subsection{Power Flow Changes}\label{subsec:property}

In this subsection, we present couple of examples in order to demonstrate complexity of power flow analysis after the \emph{primary controller's response to a MAD attack}.

First, as can be seen in Fig.~\ref{fig:Example} the relationship between the power flow changes on the lines and the demand changes are not intuitive. For example, flow on line $\{2,3\}$ is maximized when only the demand at node 3 increases (Fig.~\ref{fig:exp3}), whereas when demands at both nodes 1 and 3 increase, flow on line $\{2,3\}$ increases less (Fig.~\ref{fig:exp4}).

Another important factor affecting the amount of power flow changes on the lines, is amount of spinning reserve at each generator. For example, as can be seen in Fig.~\ref{fig:Example2}, an increase in the demand at node 1 by 3 units may result in power flow \emph{decrease} on line $\{2,3\}$, if all the generators have enough spinning reserves (Fig.~\ref{fig:exp21}). The same scenario, however, results in power flow \emph{increase} on line $\{2,3\}$, if only generators 2 and 4 have spinning reserves (Fig.~\ref{fig:exp22}). 

Fig.~\ref{fig:Example_Flows} presents the relationship between power flow changes on lines $\{2,3\}$ and $\{5,3\}$ versus power demand increase at node 1 during two different spinning reserve availability scenarios in the grid shown in Fig.~\ref{fig:exp1}. As can be seen in Fig.~\ref{subfig:flow_spin_nolimit}, if all generators have enough spinning reserve the power flows change monotonically with the demand change. However, as can be seen in Fig.~\ref{subfig:flow_spin_limit}, limited spinning reserve at generator 5 results in nonlinear relationship between the power flows and the demand change.

Following the examples provided in this subsection, it is clear that power flow changes on the lines after a MAD attack highly depend on the initial operating point of the grid and is a nonlinear problem in most cases. Despite the difficulties, however, in the next two subsections we provide efficient algorithms for finding efficient and robust operating points for the generators.

\subsection{SAFE Algorithm}\label{subsec:primary_nolimit}
 In order to avoid line overloads after the primary control response to a potential MAD attack, the grid operator needs to compute the maximum possible power flow changes on the lines following an attack (based on $\overline{\Delta p_{dl}}$ values) and enforce the power flows on the lines in OPF to be below their capacity minus the maximum possible changes. As shown in previous subsection, however, the maximum power flow changes on the lines depend on the operating point of the generators and their spinning reserve. Therefore, one cannot compute the maximum power flow changes on the lines independent of the operating points to be used in the OPF problem. 


 One way to circumvent this problem, is to enforce all the generators to have enough spinning reserves to keep the relationship between the power flow changes and demand changes linear (as in Fig.~\ref{subfig:flow_spin_nolimit}), and use this linear relationship to compute the maximum power flow changes on the lines based on the operating point of the generators. These values can then be added to the OPF problem without making the problem nonlinear and nonconvex.

For each load $i$, define $\vec{v}_i=[v_{i1},v_{i2},\dots,v_{in}]^T$ to denote the primary controllers' response to a unit demand increase at load $i$. If all generators have enough spinning reserve, each generator $j$ will increase its generation by $v_{ij}:=(1/R_j)/(\sum_{l=1}^n1/R_l)$ to compensate for a unit demand increase at node $i$ (as described in Section~\ref{subsec:frequency}). Hence, by defining $\vec{w}_i :=\vec{v}_i-\vec{e}_i$ (recall from Section~\ref{sec:model} that $\vec{e_i}$ is the i$^{\text{th}}$ fundamental basis of $\mathbb{R}^n$) one can compute the change in flow of line $e=(i,j)$ solely in terms of changes in the demands ($\Delta p_{di}$s):
\begin{align}\label{eq:delta_f}
\Delta f_{ij} = 1/x_{ij} (\A_i^+-\A_j^+) \sum_{l=1}^n \Delta p_{dl}  \vec{w}_l.
\end{align}
 Recall that $-\overline{\Delta p_{dl}}\leq \Delta p_{dl}\leq \overline{\Delta p_{dl}}$ based on the grid operator's estimation of the adversary's power. 
Hence, the maximum flow change on line $(i,j)$ can be computed using (\ref{eq:delta_f}) as:
\begin{equation}\label{eq:max_delta_f}
\Delta f_{ij}^{\max}= 1/x_{ij} \sum_{l=1}^n \overline{\Delta p_{dl}} |(\A_i^+-\A_j^+) \vec{w}_l|,
\end{equation}
since for each $l$, $\Delta p_{dl}$ can be selected by the adversary to be equal to $-\overline{\Delta p_{dl}}$, if $(\A_i^+-\A_j^+)^T \vec{w}_l<0$, and equal to $\overline{\Delta p_{dl}}$, if $(\A_i^+-\A_j^+) \vec{w}_l\geq0$.
Now, to ensure that no lines are overloaded after a MAD attack, all the system operator needs to do is to replace the capacity of each line $(i,j)$ in the OPF problem  by $\overline{f_{ij}}-\Delta f_{ij}^{\max}$. The only other constraint that needs to be added to the OPF problem is to make sure that each generator $i$ with $0<1/R_i$ has enough spinning reserve to increase its generation according to its governor droop. For this, define $\overline{S_{\Delta p_d}}:=\sum_{l=1}^n \overline{\Delta p_{dl}}$. Hence, each generator's operating point should be within following limits:
\begin{equation}\label{eq:pg_limits}
 \forall1\leq i\leq n:\underline{p_{gi}}+\frac{1/R_i}{\sum_{l=1}^n 1/R_l}\overline{S_{\Delta p_d}} \leq p_{gi} \leq \overline{p_{gi}}-\frac{1/R_i}{\sum_{l=1}^n 1/R_l} \overline{S_{\Delta p_d}}.
\end{equation}

Therefore, the robust OPF problem can be written as follows:
\begin{align}
\label{eq:ec_dis}&\min_{\tet,\vec{f},\vec{p_g}}& &\sum_{l=1}^n c_l(p_{gl}),\\
&\qquad\text{s.t.}& &(\ref{eq:p}),(\ref{eqn:flow1}),(\ref{eqn:flow2}), (\ref{eq:max_delta_f}), (\ref{eq:pg_limits}),\nonumber\\
&& &|f_{ij}|\leq \overline{f_{ij}}-\Delta f_{ij}^{\max}, \quad \forall(i,j)\in E\nonumber\\
&&&\vec{p}=\vec{p_g}-\vec{p_d}.\nonumber
\end{align}

We call the algorithm for finding a robust operating point for generators based on solving (\ref{eq:ec_dis}), the \underline{S}ecuring \underline{A}dditional margin \underline{F}or generators in \underline{E}conomic dispatch (SAFE) Algorithm. Since this algorithm limits the operating points of the generators, it is obvious that it may not obtain the \emph{minimum cost} robust operating points for the generators. In the next subsection, we provide an algorithm, albeit computationally more expensive, for finding robust operating points for the generators without limiting their generation.
\subsection{IMMUNE Algorithm}\label{subsec:primary_limit}
In (\ref{eq:delta_f}), we assumed that none of the generators reach their maximum/minimum capacity as they increase/decrease their generation according to their droop characteristics. 
However, by allowing some generators to reach their maximum/minimum capacity, one may find robust operating points for the generators with a lower cost.

In this subsection, we assume without loss of generality that the total demand change $S_{\Delta p_d}:=\sum_{i=1}^{d}\Delta p_{di}$ is positive. Hence, we focus mainly on generators' maximum capacity. However, same results hold for the case that $S_{\Delta p_d}<0$ as well.


Once a generator reaches its maximum capacity, it cannot increase its generation anymore, and therefore other generators should generate more to compensate for the extra demand. The following lemma provides the amount each generator generates based on its spinning reserve and governor droop characteristic to compensate for the extra demand after a MAD attack. 
\vspace*{0.2cm}
\begin{lemma}\label{lem:gen_order}
Suppose generators are ordered such that if $i<j$, $R_{i}(\overline{p_{gi}}-p_{gi})\leq R_j (\overline{p_{gj}}-p_{gj})$. Define $t_i:=R_{i}(\overline{p_{gi}}-p_{gi})$ and $S_i:=\sum_{l=1}^i t_l/R_l+\sum_{l=i+1}^n t_i/R_l$. If $S_i<S_{\Delta p_d}\leq S_{i+1}$, to compensate for the extra demand, generators 1 to $i$ reach their maximum capacity and each generator $j>i$ generates $\frac{1/R_{j}}{\sum_{l=i+1}^n 1/R_l} \big(S_{\Delta p_d}-\sum_{l=1}^i (\overline{p_{gl}}-p_{gl})\big)$.
%
\end{lemma}


In general, as demonstrated in Figs.~\ref{fig:Example2} and \ref{fig:Example_Flows}, due to power generation limits, power flow on a line may not change monotonically as demand changes in a specific node--as in (\ref{eq:delta_f}). 
Hence,  the maximum change in the power flows cannot be found in a closed form as in (\ref{eq:max_delta_f}). However, one may be able to find an upper bound on the maximum power flow change on a line.

Upper bounds on the maximum power flow changes after a MAD attack can be computed by assuming the worst case initial operating points and also assuming that generators can be arbitrarily assigned to provide extra required generation. 
In particular, an upper bound $\widehat{\Delta f_{ij}}$ for the power flow changes on line $(i,j)$ can be computed as:
\begin{align}
\widehat{\Delta f_{ij}} := &\max_{\vec{p_g},\vec{\Delta p_d},\vec{\Delta p_g}}&&\!\!\!\left|1/x_{ij}(\A_i^+-\A_j^+)(\vec{\Delta p_g}-\vec{\Delta p_d})\right|\label{eq:upperbound1}\\
&\qquad\text{s.t.}&&\vec{1}^T(\vec{p_g}-\vec{p_d})=0,\nonumber\\
&&&\vec{1}^T(\vec{\Delta p_g}-\vec{\Delta p_d})=0,\nonumber\\
&&& -\overline{\Delta p_{dl}}\leq \Delta p_{dl}\leq \overline{\Delta p_{dl}}, \quad 1\leq l\leq n \nonumber\\
&&& \underline{p_g}\leq \vec{p_g}\leq\overline{p_g},\nonumber\\
&&& 0\leq\Delta p_{gl}\leq \overline{p_{gl}}-p_{gl}, \quad  1\leq l\leq n,\nonumber\\
&&& S_{\Delta p_d}\geq 0.\nonumber
\end{align}


Optimization (\ref{eq:upperbound1}) is a  Linear Program (LP) that can be solved efficiently for each line $(i,j)$. Using these upper bounds, the solution to the following OPF problem provides robust operating points for the generators: 
\begin{align}
\label{eq:ec_dis2}
&\min_{\tet,\vec{f},\vec{p_g}}& &\sum_{l=1}^n c_l(p_{gl}),\\
&\qquad\text{s.t.}& &(\ref{eq:p}),(\ref{eqn:flow1}),(\ref{eqn:flow2}), (\ref{eq:pg}),\nonumber\\
&& & |f_{ij}|\leq \overline{f_{ij}}-\widehat{\Delta f_{ij}},\quad\forall(i,j)\in E\nonumber\\
&& &\vec{p}=\vec{p_g}-\vec{p_d}.\nonumber
\end{align}

It is obvious that enforcing the power flows on all the lines, such as $(i,j)$, to be less than $\overline{f_{ij}}-\widehat{\Delta f_{ij}}$ in the OPF problem as in (\ref{eq:ec_dis2}) ensures that none of the lines will be overloaded after a potential MAD attack. However, the solution to (\ref{eq:ec_dis2}) may not provide the optimal robust operating points for the generators. To achieve a better solution, we introduce an iterative algorithm that solves the OPF problem and updates the line capacities to ensure robustness. We will then use the upper bounds $\widehat{\Delta f_{ij}}$s to prove that the algorithm will converge to a local optimal solution.

First, given the operating points $p_{g1},\dots,p_{gn}$ to the OPF problem, the maximum power flow change on line $(i,j)$ (denoted by $\Delta f_{ij}^{\max}$) can be computed by solving the following optimization problem:
\begin{align}
\label{eq:max_delta_f_limit}&\Delta f_{ij}^{\max} = &&\max_{\vec{\Delta p_d}} &&\sgn(f_{ij}) \Big(1/x_{ij} \sum_{l=1}^n -\Delta p_{dl} (a_{il}^+-a_{jl}^+)\\
&&&&&+1/x_{ij}\sum_{l=1}^n f_l(S_{\Delta p_d}) (a_{il}^+-a_{jl}^+)\Big)\nonumber\\
&&&\qquad\text{s.t.}&& -\overline{\Delta p_{dl}}\leq\Delta p_{dl}\leq \overline{\Delta p_{dl}},\quad 1\leq l\leq n\nonumber\\
&&&&& ~S_{\Delta p_d}\geq 0.\nonumber
\end{align}
in which $f_l(.)$s denote piecewise linear functions that determine the extra output of the generators based on the total demand change $S_{\Delta p_d}$. Since we assumed that $p_{g1},\dots, p_{gn}$ are given, functions $f_l(.)$ can be uniquely determined using Lemma~\ref{lem:gen_order}. Notice that $\sgn(f_{ij})$ in the objective of (\ref{eq:max_delta_f_limit}) is to ensure that the maximum changes are in the direction of \emph{increase} in the power flow on line $(i,j)$. Hence, for all lines $\Delta f_{ij}^{\max}\geq0$.

\begin{lemma}\label{lem:max_f_polynomial}
Optimization (\ref{eq:max_delta_f_limit}) can be solved in polynomial time for each $(i,j)\in E$.
\end{lemma}
\begin{proof}
Without loss of generality, assume that generators are ordered such that $t_1\leq t_{2}\leq\dots\leq t_n$ as defined in Lemma~\ref{lem:gen_order}. It is easy to see that by using Lemma~\ref{lem:gen_order} and defining $S_0:=0$, one can solve (\ref{eq:max_delta_f_limit}) in different linear regions of $f_l(.)$s by considering additional conditions for $S_{\Delta p_d}$ (for $0\leq z< n$):
\begin{eqnarray}\label{eq:condition}
S_z\leq S_{\Delta p_d}< S_{z+1}.
\end{eqnarray}
Under condition (\ref{eq:condition}), $f_{l}(.)$s can be determined as follows:
\begin{equation}
f_l(S_{\Delta p_d}) =
\begin{cases}
\overline{p_l}-p_l& l\leq z,\\
\frac{1/R_{l}\big(S_{\Delta p_d}-\sum_{w=1}^z (\overline{p_w}-p_{w})\big)}{\sum_{w=z+1}^n 1/R_w} & l>z.
\end{cases}
\end{equation}
 Hence, all the $f_{l}(.)$ are either constant or linear functions in (\ref{eq:max_delta_f_limit}) and therefore (\ref{eq:max_delta_f_limit}) can be solved efficiently using LP. Hence, by solving (\ref{eq:max_delta_f_limit}) at most $n$ times (once for every condition (\ref{eq:condition}) for different $z$) $\Delta f_{ij}^{\max}$ can be found in polynomial time.
 \end{proof}
After computing $\Delta f_{ij}^{\max}$ values, one can use them to verify if any of the lines will be overloaded after an attack. If yes, then add additional constraints to the OPF problem to ensure that those lines will not be overloaded. The OPF problem can then be solved with new additional constraints and the process continues until no additional constraints are needed. We call this algorithm \underline{I}teratively \underline{M}ini\underline{M}ize and bo\underline{UN}d \underline{E}conomic dispatch (IMMUNE) Algorithm (summarized in Algorithm~\ref{algorithm:IMMUNE}).

\begin{algorithm}[t]
\caption{\underline{I}teratively \underline{M}ini\underline{M}ize and bo\underline{UN}d \underline{E}conomic dispatch (IMMUNE)}
\label{algorithm:IMMUNE}
\small
\begin{trivlist}
\item\textbf{Input:} $G$
\end{trivlist}
\begin{algorithmic}[1]
\STATE flag = 1
\STATE Define $c_{ij}:=\overline{f_{ij}}$ for all $(i,j)\in E$
\WHILE{flag}
    \STATE Compute OPF problem (\ref{eq:dcopf}) such that $\forall (i,j)\in E:|f_{ij}|\leq c_{ij}$
    \IF{OPF is not feasible}
    \STATE \textbf{return} none
    \ENDIF
    \STATE Compute $\Delta f_{ij}^{\max}$ by solving (\ref{eq:max_delta_f_limit}) for all $(i,j)\in E$
    \STATE flag = 0
    \FOR{$(i,j)\in E$}
    \vspace*{0.15cm}
        \IF{$\overline{f_{ij}}<|f_{ij}|+\Delta f_{ij}^{\max}$}
            \vspace*{0.15cm}
            \STATE $c_{ij}=\overline{f_{ij}}-\Delta f_{ij}^{\max}$
            \STATE flag = 1
        \ENDIF
    \ENDFOR
\ENDWHILE
\STATE \textbf{return} $p_{g1},p_{g2},\dots,p_{gn}$
\end{algorithmic}
\end{algorithm}

\begin{lemma}\label{lem:IMMUNE_converge}
If (\ref{eq:ec_dis2}) is feasible, then the IMMUNE Algorithm converges to a local optimum solution.
\end{lemma}
%

Lemma~\ref{lem:IMMUNE_converge} provides a sufficient condition such that the IMMUNE Algorithm converges to a local optimum. However, even if (\ref{eq:ec_dis2}) is not feasible, the system operator can still run the IMMUNE Algorithm to obtain a local optimum solution if the OPF problem remains feasible at each iteration of the algorithm.

We can also provide an upper bound on the number of iterations that IMMUNE algorithm requires to converge. For this reason, the algorithm needs to change discrete changes to the capacities at each iteration.

\begin{lemma}\label{lem:immune_runtime}
If the IMMUNE Algorithm changes $c_{ij}$ at each iteration by a discrete amount such as $c_{ij}=\max\{\lfloor\overline{f_{ij}}-\Delta f_{ij}^{\max}\rfloor,\overline{f_{ij}}-\widehat{\Delta f_{ij}}\}$, then it 
terminates in at most $O(\sum_{(i,j)\in E} \lceil \widehat{\Delta f_{ij}}\rceil)$ iterations.
\end{lemma}

Following a similar idea, one can increase the running time of the IMMUNE algorithm by applying more aggressive update rules for the capacities in line 11 of the algorithm. For example, line 11 can be replaced by $c_{ij}=0.9(\overline{f_{ij}}-\Delta f_{ij}^{\max})$ or $c_{ij}=0.95(\overline{f_{ij}}-\Delta f_{ij}^{\max})$. We call these variations of the IMMUNE Algorithm, IMMUNE-0.9 and IMMUNE-0.95. In Section~\ref{subsec:prime_num}, we numerically evaluate and compare the performance of these algorithms and demonstrate that more aggressive update rules result in faster convergence. 

One favorable property of the IMMUNE Algorithm is that it can be easily parallelized. This parallelization can be used to simultaneously compute $\Delta f_{ij}^{\max}$ for all the lines at each iteration  in order to expedite the algorithm.

If the OPF problem becomes infeasible in any iteration of the IMMUNE Algorithm, there are two ways to circumvent the issue: (i) By considering higher temporary limits for the lines (e.g., $1.1\overline{f_{ij}}$) which is a common practice in power systems operation, but the operator needs to ensure that line overloads can be cleared during the secondary control, or (ii) by returning to the unit commitment problem and change the list of committed generators to make sure (\ref{eq:ec_dis2}) is feasible. We will address the first approach in the next section in detail. However, the second approach is beyond the scope of this paper and is part of our future work.

\section{Power Flows: Secondary Control}\label{sec:secondary}

\begin{figure}[t]
\centering
\begin{subfigure}{0.20\textwidth}
\centering
\includegraphics[scale=0.43]{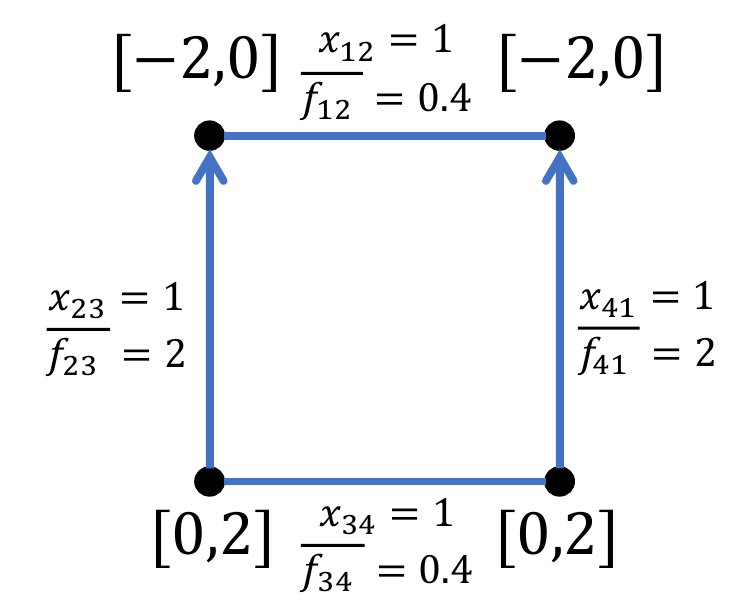}
\caption{}
\end{subfigure}
\begin{subfigure}{0.13\textwidth}
\centering
\includegraphics[scale=0.43]{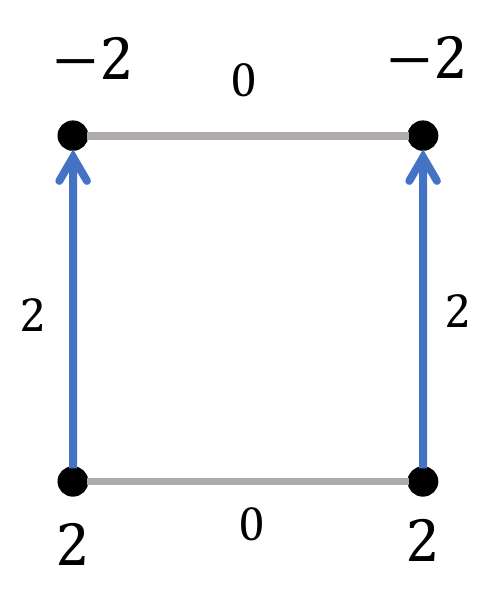}
\caption{}
\label{subfig:Sec_Example_2}
\end{subfigure}
\begin{subfigure}{0.13\textwidth}
\centering
\includegraphics[scale=0.43]{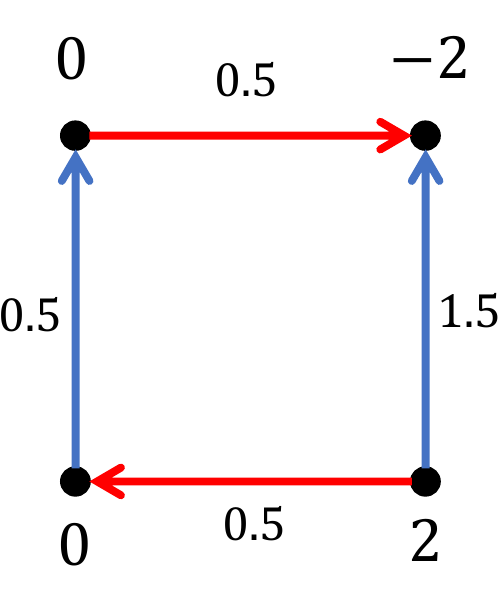}
\caption{}
\label{subfig:Sec_Example_3}
\end{subfigure}
\caption{Complexity of secondary controller problem. (a) Secondary controller problem setting, (b) an attack that maximizes the demand, and (c) an attack that minimizes the demand at one node and maximizes the demand at another node.}
\label{fig:Sec_Example}
\end{figure}

 In cases that primary control cannot prevent line overloads, 
 the system operator have to clear these overloads during the secondary control instead. 
 In such cases, the operator needs to make sure in advance that after the primary control's response to a MAD attack, there are operating points for the generators such that the demand can be supplied with no line overloads (i.e., secondary controller can clear the overloads). 
 Assuming that the maximum and minimum reachable demand at node $i$ by an adversary is $\underline{p_{di}}$ and $\overline{p_{di}}$, respectively, this problem can be defined as the \emph{secondary controller problem}:

\vspace*{0.2cm}
\emph{Secondary controller problem:} For any $p_{d1},p_{d2},\dots,p_{dn}$ that $\forall 1\leq i\leq n:\underline{p_{di}}\leq p_{di}\leq \overline{p_{di}}$, are there operating points $p_{g1},\dots,p_{gn}$ for the generators such that $\forall 1\leq i\leq n:\underline{p_{gi}}\leq p_{gi}\leq \overline{p_{gi}}$, $\vec{1}^T(\vec{p_g}-\vec{p_d})=0$, and no lines are overloaded?
\begin{definition}
A grid is called \emph{secondary controllable} if the answer to the secondary controller problem is yes.
\end{definition}
Notice that \emph{operating cost of the generators are not important during the secondary control,} since the secondary controller activates only after a potential attack and the operator needs to bring back the grid to its normal state as soon as possible at any cost. 
Fig.~\ref{fig:Sec_Example} provides an example of the secondary controller problem. As can be seen in Fig.~\ref{subfig:Sec_Example_2}, when the demands are all equal to their maximum level after a MAD attack, the demand can be supplied by generators with no line overloads. However, as presented in Fig.~\ref{subfig:Sec_Example_3}, when the demand is increased to it maximum level at one node and decreased to its minimum at another one, there is no possible way to supply the demand such that no lines are overloaded. This example clearly evinces that the secondary controller problem is not intuitive.

In the following subsections, we study the secondary controller problem in detail and provide efficient algorithms to verify the secondary controllability of a power system. 
\subsection{Maxmin Formulation}\label{subsec:maxmin}
 One way of verifying the secondary controllability of a power system is by exploiting \emph{linear bilevel programs}~\cite{falk1973linear,bard1998practical}. The secondary controller problem can be written in the form of a max-min linear problem which is a special form of \emph{linear bilevel programs} as follows:
 \begin{align}\label{eq:maxmin}
&\max_{\vec{p_d}}\min_{\vec{p_g},\vec{q},\vec{f},\tet}&&\vec{1}^T\vec{q}\\
&\quad\quad\quad\text{s.t.}&& (\ref{eq:p}),(\ref{eqn:flow1}),(\ref{eqn:flow2}), (\ref{eq:f}), (\ref{eq:pg}),\nonumber\\
&&&\vec{p}=\vec{p_{g}}-\vec{p_{d}}+\vec{q}, \nonumber\\
&&& q_i\geq 0, \quad 1\leq i\leq n\nonumber\\
&&& \underline{p_{di}}\leq p_{di}\leq \overline{p_{di}}, \quad 1\leq i\leq n.\nonumber
 \end{align}
In optimization problem (\ref{eq:maxmin}), vector $\vec{p_d}$ should be selected such that for the best possible selection of vector $\vec{p_g}$ and positive auxiliary vector $\vec{q}$, the objective value is maximized. Following proposition relates the solution of (\ref{eq:maxmin}) to the secondary controller problem.

\begin{proposition}\label{prop:maxmin}
The optimal solution of (\ref{eq:maxmin}) is 0 if, and only if, the grid is secondary controllable.
\end{proposition}
\begin{proof}
If the optimal solution to (\ref{eq:maxmin}) is 0, then for any demand vector $\vec{p_d}$, vector of generation values $\vec{p_g}$ can be selected such that $\vec{1}^T(\vec{p_g}-\vec{p_d})=0$ and no lines are overloaded. Hence, the grid is secondary controllable. Now if the grid is secondary controllable, then for all demand vectors $\vec{p_d}$, there exists a vector of generation $\vec{p_g}$ such that $\vec{1}^T(\vec{p_d}-\vec{p_d})=0$ and no lines are overloaded. Hence, the auxiliary vector $\vec{q}$ can be selected to be equal to 0 by the minimization part of (\ref{eq:maxmin}) for any vector $\vec{p_d}$. Therefore, the optimal solution to (\ref{eq:maxmin}) would be 0.
\end{proof}

Proposition~\ref{prop:maxmin} clearly demonstrates that solving (\ref{eq:maxmin}) can determine secondary controllability of a power system. Moreover, when optimal solution of (\ref{eq:maxmin}) is greater than 0, the nonzero entries of the optimal vector $\vec{q}$ can reveal the minimum extra generation required to ensure secondary controllability of the system.

Despite many advantages of formulation (\ref{eq:maxmin}), the max-min linear program is nonconvex~\cite{bialas1984two} and proved to be NP-hard~\cite{hansen1992new}. 
Therefore existing efficient algorithms for solving (\ref{eq:maxmin}) only obtain local optimal solutions~\cite{bard1998practical}. However, a local optimal solution of (\ref{eq:maxmin}) with value 0 does not guarantee the secondary controllability of the system since the optimal solution may not be zero.

One way of solving (\ref{eq:maxmin}) optimally, albeit in exponential running time, is through brute force search. Following lemma demonstrates that to solve the secondary controller problem, one needs to check only the extreme demand points due to convexity of the space of all possible demand values and linearity of power flow equations.

\begin{lemma}\label{lem:extreme}
The grid is secondary controllable, if and only if for all $p_{d1},\dots,p_{dn}$ such that  $p_{di}\in\{\overline{p_{di}},\underline{p_{di}}\}$ there exist operating points $p_{g1},\dots,p_{gn}$ for the generators such that $\forall 1\leq i\leq n:\underline{p_{gi}}\leq p_{gi}\leq \overline{p_{gi}}$, $\vec{1}^T(\vec{p_g}-\vec{p_d})=0$, and no lines are overloaded.
\end{lemma}
One the other hand, for a given demand vector $\vec{p_d}$, it can be verified in polynomial time whether there exist operating points for the generators that satisfy the secondary controller problem by solving the minimization part of (\ref{eq:maxmin}) using LP:
\begin{align}\label{eq:minPg}
&\min_{\vec{p_g},\vec{q},\vec{f},\tet}&&\vec{1}^T\vec{q}\\
&\qquad\text{s.t.}& &(\ref{eq:p}),(\ref{eqn:flow1}),(\ref{eqn:flow2}), (\ref{eq:f}), (\ref{eq:pg}),\nonumber\\
&&&\vec{p}=\vec{p_g}-\vec{p_d}+\vec{q}\nonumber\\
&&& q_i\geq 0, \quad 1\leq i\leq n.\nonumber
\end{align}

If the optimum solution to (\ref{eq:minPg}) is not 0, then the optimal vector $\vec{q}$ can be used by the operator to make more generators online for controllability of the grid. Hence by solving (\ref{eq:minPg}) for all extreme demand vectors, one can verify secondary controllability of a system in \emph{exponential running time} and also find how to make it controllable--if it is not--based on obtained vectors $\vec{q}$.

By focusing only on nodes with the largest demands, one can approximately verify if for a subset of  extreme points there exist operating points for the generators satisfying the secondary controller problem. In general, however, such approach may not be able to guarantee the secondary controllability of a grid. Hence, in the next subsection, we provide \emph{sufficient conditions} to ensure secondary controllability of a grid in polynomial time. 

\subsection{Predetermined Secondary Controllers}\label{subsec:secondary_predetermined}

Despite the difficulty in exact determination of secondary controllability of a grid, in this subsection, we introduce and exploit suboptimal predetermined controllers to verify controllability of a grid \emph{with no false positives} (i.e., presented methods cannot determine \emph{uncontrollability} of a system).


In order to verify secondary controllability of the grid, one can find \emph{the best} predetermined way to set the generation values given a demand vector $\vec{p_d}$ such that the maximum power flows over all demand vectors is minimized. In particular, we define the $\vec{\beta}$-determined controller as follows.

\begin{definition}[$\vec{\beta}$-determined controller] \label{def:beta_determined} For any demand vector $\vec{p_d}$, set $\vec{p_g}=(\sum_{i=1}^n p_{di})\times \vec{\beta}$, for a vector $\vec{\beta}$ satisfying: 
\begin{itemize}
\item[(i)] $\vec{\beta}\geq 0$, \quad (ii) $\vec{1}^T\vec{\beta}=1$, \quad (iii) $(\sum_{i=1}^n \overline{p_{di}})\times \vec{\beta}\leq \overline{p_g}$,
\item[(iv)] $(\sum_{i=1}^n \underline{p_{di}})\times \vec{\beta}\geq \underline{p_g}$.
\end{itemize}
\end{definition}
\begin{definition}\label{def:reliable}
A controller is called \emph{reliable}, if for all feasible demand vectors $\vec{p_d}$, it provides a vector of operating points for the generators like $\vec{p_g}$ such that $|\vec{f}|=|\B(\vec{p_g}-\vec{p_d})|\leq \overline{f}$.
\end{definition}


\begin{proposition}\label{prop:second_approx_beta}
If there exists a vector $\vec{\beta}$ such that the $\vec{\beta}$-determined controller is reliable, then the grid is secondary controllable.
\end{proposition}

For a vector $\vec{\beta}$ satisfying conditions (i-iv) in Definition~\ref{def:beta_determined}, define vectors $\vec{w_i}^{(\beta)}:=-\vec{e_i}+ \vec{\beta}$ for $1\leq i\leq n$  (as in Section~\ref{subsec:primary_nolimit}). The following lemma proves that maximum flow on the lines over all feasible demand vectors, given a $\vec{\beta}$-determined controller, can deterministically be computed.
\begin{lemma}\label{lem:maxflow_beta}
Given a $\vec{\beta}$-determined controller, the maximum power flow on each line $e_k$ over all possible demand vectors is:
\begin{equation}\label{eq:maxflow_beta}
\max_{\underline{p_d}\leq\vec{p_d}\leq \overline{p_d}} |f_k| = \left|\sum_{i=1}^n \frac{(\overline{p_{di}}+\underline{p_{di}})}{2} \B_k\vec{w_i}^{(\beta)}\right| + \sum_{i=1}^n \frac{(\overline{p_{di}}-\underline{p_{di}})}{2} |\B_k\vec{w_i}^{(\beta)}|.
\end{equation}
\end{lemma}
The main question is now whether there exists a vector $\vec{\beta}$ such that the maximum power flows as determined in (\ref{eq:maxflow_beta}) are less than their capacities? We prove that one can examine this efficiently and in polynomial time by solving the following optimization:
\begin{align}\label{eq:optimalBeta}
&\min_{\eta,\vec{\beta},\vec{f}}&&\eta\\
&\quad\text{s.t.}&& \text{(i-iv) in Definition~\ref{def:beta_determined}},\nonumber\\
&&& \vec{f} = |\B\W^{(\beta)}(\overline{p_d}+\underline{p_d})/2| + |\B\W^{(\beta)}|(\overline{p_d}-\underline{p_d})/2,\nonumber\\
&&& \vec{f}\leq \eta \overline{f},\nonumber
\end{align}

in which matrix $\W^{(\beta)}:= [\vec{w_1}^{(\beta)},\dots,\vec{w_n}^{(\beta)}]$. The following proposition demonstrates that (\ref{eq:optimalBeta}) can be solved using LP in polynomial time. Moreover, it indicates that the optimal solution to (\ref{eq:optimalBeta}) can provide the best vector $\vec{\beta}$ for deterministically controlling the grid and its optimal value demonstrates if the corresponding $\vec{\beta}$-determined controller is reliable. 


\begin{proposition}\label{prop:LPBeta}
Optimization (\ref{eq:optimalBeta}) can be solved using LP. Moreover, if the optimal value $\eta^*$ to (\ref{eq:optimalBeta}) is less than or equal to 1, then the $\vec{\beta}^*$-determined controller obtained from corresponding solution is reliable, and therefore the grid is secondary controllable.
\end{proposition}


From (\ref{eq:maxflow_beta}), it can be seen that the formula for computing maximum flow on the lines consists of two separate sums which can be controlled by different vectors and obtained a better controller. Hence, one can define the $(\vec{\gamma},\vec{\beta})$-determined controller as follows.

\begin{definition}[$(\vec{\gamma},\vec{\beta})$-determined controller]\label{def:gam_beta_determined} For any demand vector $\vec{p_d}$, set $\vec{p_g}=(\sum_{i=1}^n (\overline{p_{di}}+\underline{p_{di}})/2)\times \vec{\gamma}+(\sum_{i=1}^n (p_{di}-\overline{p_{di}}/2-\underline{p_{di}}/2))\times \vec{\beta}$, for vectors $\vec{\gamma}$ and $\vec{\beta}$ satisfying:
\begin{itemize}
\item[(i)] $\vec{\beta},\vec{\gamma}\geq 0$, \quad (ii) $\vec{1}^T\vec{\gamma}=\vec{1}^T\vec{\beta}=1$,
\item[(iii)] $(\sum_{i=1}^n (\overline{p_{di}}+\underline{p_{di}})/2)\times \vec{\gamma}+(\sum_{i=1}^n (\overline{p_{di}}-\underline{p_{di}})/2)\times \vec{\beta}\leq \overline{p_g}$,
\item[(iv)] $(\sum_{i=1}^n (\overline{p_{di}}+\underline{p_{di}})/2)\times \vec{\gamma}+(\sum_{i=1}^n (-\overline{p_{di}}+\underline{p_{di}})/2)\times \vec{\beta}\geq \underline{p_g}$.
\end{itemize}
\end{definition}

The $(\vec{\gamma},\vec{\beta})$-determined controller generalizes the $\vec{\beta}$-determined controller (just set $\vec{\gamma}=\vec{\beta}$) and
it is easy to see that the maximum power flow on the lines over all demand vectors, given a $(\vec{\gamma},\vec{\beta})$-determined controller can be computed similar to (\ref{eq:maxflow_beta}) as follows:
\begin{equation}\label{eq:maxflow_beta_gamma}
\max_{\underline{p_d}\leq\vec{p_d}\leq \overline{p_d}} |f_k| = \left|\sum_{i=1}^n \frac{(\overline{p_{di}}+\underline{p_{di}})}{2} \B_k\vec{w_i}^{(\gamma)}\right| + \sum_{i=1}^n \frac{(\overline{p_{di}}-\underline{p_{di}})}{2} |\B_k\vec{w_i}^{(\beta)}|.
\end{equation}

Optimal $(\vec{\gamma},\vec{\beta})$-determined controller can be found similar to the optimal $\vec{\beta}$-determined controller using an optimization similar to (\ref{eq:optimalBeta}) with a few small changes: 
\begin{align}\label{eq:optimalGamBeta}
&\min_{\eta,\vec{\gamma},\vec{\beta},\vec{f}}&&\eta\\
&\quad\text{s.t.}&& \text{(i-iv) in Definition~\ref{def:gam_beta_determined}},\nonumber\\
&&&  \vec{f} = |\B\W^{(\gamma)}(\overline{p_d}+\underline{p_d})/2| + |\B\W^{(\beta)}|(\overline{p_d}-\underline{p_d})/2,\nonumber\\
&&& \vec{f}\leq \eta \overline{f}.\nonumber
\end{align}

Again, as in the $\vec{\beta}$-determined controller case, the optimal value of (\ref{eq:optimalGamBeta}) determines if the optimal $(\vec{\gamma},\vec{\beta})$-determined controller is reliable or not. Hence, the grid operator can use (\ref{eq:optimalGamBeta}) to efficiently determine the secondary controllability of the grid, albeit obtaining false negatives in some cases. 

In Section~\ref{sec:numerical}, we numerically evaluate the performance of the controllers introduced in this section. Before that, however, we demonstrate that these controllers can be used to efficiently provide lower bounds on the scale of MAD attacks that a grid can tolerate.

\section{$\alpha D$-robustness}\label{sec:contingency}
Power grids are required to withstand single equipment failures (e.g., lines, generators, and transformers) with no interruptions in their operation (a.k.a. $N-1$ standard)~\cite{wood2012power}. Following $N-1$ standard, we define a new standard for the grid operation to  ensure its robustness against MAD attacks. We call this $\alpha D$ standard that requires grids to be robust against contingencies resulted by changing the demand by $\alpha$ fraction. We call a grid that is robust against these type of contingencies, $\alpha D$-robust.

The $\alpha D$-robustness can be enforced during the economic dispatch using SAFE or IMMUNE algorithms developed in Section~\ref{sec:primary} so that no lines are overloaded after the primary control.
In this section, however, assuming that the power lines can handle temporary overloads (as in the previous section), we are interested in finding maximum $\alpha$ such that the grid is $\alpha D$-robust \emph{after the secondary control}. We denote that value by $\alpha^{\max}$. 

Since as we described in the previous section, verifying the secondary controllability of the grid for a given upper and lower limits on the demands is hard, we cannot expect to find the $\alpha^{\max}$ efficiently. Hence, in the next two subsections, we develop efficient methods for obtaining upper and lower bounds on $\alpha^{\max}$.

%
%
\subsection{Upper Bound}
Assume $\vec{p_d}^\dagger$ denotes the vector of predicted demand values. For a given $\alpha$, the demand vector $\vec{p_d}$ resulted by a MAD attack will be bounded by $(1-\alpha)\vec{p_d}^\dagger\leq \vec{p_d}\leq (1+\alpha)\vec{p_d}^\dagger$. Now if a grid is $\alpha D$-robust, it should particularly be robust against the maximum demand attack. Hence, finding the maximum $\alpha$ for which the grid can handle the maximum demand attack provides an upper bound for $\alpha^{\max}$.  Such $\alpha$ can be found efficiently by a LP:
 \begin{align}\label{eq:max_alpha}
&\max_{\alpha,\vec{p_d},\vec{p_g},\vec{f},\tet}&&\alpha\\
&\quad\quad\quad\text{s.t.}& &(\ref{eq:p}),(\ref{eqn:flow1}),(\ref{eqn:flow2}), (\ref{eq:f}), (\ref{eq:pg}),\nonumber\\
&&&\vec{p_d}=(1+\alpha) p_d^\dagger, \nonumber\\
&&& \vec{p} = (\vec{p_g}-\vec{p_d}).\nonumber
 \end{align}
\begin{proposition}\label{prop:alpha_hat}
Assume $\hat{\alpha}$ denotes the optimal value of (\ref{eq:max_alpha}), then $\alpha^{\max}\leq \hat{\alpha}$.
\end{proposition}

The optimal value of (\ref{eq:max_alpha}) provides a good upper bound for $\alpha^{\max}$ and can be computed efficiently. In the next subsection, we provide algorithms to find lower bounds for $\alpha$ based on the controllers developed in Section~\ref{subsec:secondary_predetermined}.
\subsection{Lower Bound}
To find a lower bound for $\alpha^{\max}$, we use the controllers in Section~\ref{subsec:secondary_predetermined} to limit the secondary controller's ability in changing the generators' operating points. Limiting the secondary controller's ability allows us to efficiently approximate the maximum $\alpha$, but because of this limitation, we only obtain lower bounds for $\alpha^{\max}$. 

First, assume that we limit the secondary controller to the $\vec{\beta}$-controller for a fixed $\vec{\beta}$. We show that in this case the maximum $\alpha$ can be found by solving a single LP. Assume $\vec{p_g}^*$ is the optimal solution to (\ref{eq:max_alpha}) with value $\hat{\alpha}$ and set $\vec{\beta}=\vec{p_g}^*/\|\vec{p_g}^*\|_1$ (i.e., controller only scales down the generation compared to the maximum demand case). Using (\ref{eq:maxflow_beta}), we show that the optimal value of the following LP gives a lower bound for $\alpha^{\max}$:
\begin{align}\label{eq:max_alpha_SC_controller}
&\max_{\alpha,\vec{f}} &&\alpha\\
&\quad\text{s.t.}&& (1+\alpha)(\sum_{i=1}^n p_{di}^\dagger)\times\vec{\beta}\leq \overline{p_g},\nonumber\\
&&&  (1-\alpha)(\sum_{i=1}^n{p_{di}^\dagger})\times \vec{\beta}\geq \underline{p_g},\nonumber\\
&&& \vec{\beta}=\vec{p_g}^*/\|\vec{p_g}^*\|_1,\nonumber\\
&&& \vec{f} = |\B\W^{(\beta)}\vec{p_d}^\dagger| + |\B\W^{(\beta)}|(\alpha\vec{p_d}^\dagger),\nonumber\\
&&& |f_{ij}|\leq \overline{f_{ij}}, \quad \forall(i,j)\in E.\nonumber
\end{align}

\begin{proposition}\label{prop:alpha_star}
The optimal solution $\alpha^*$ of (\ref{eq:max_alpha_SC_controller}) can be found in polynomial time using LP. Moreover, $\alpha^*\leq \alpha^{\max}$.
\end{proposition}

Optimization (\ref{eq:max_alpha_SC_controller}) allows us to efficiently compute a lower bound for $\alpha^{\max}$. However, similar to  Section~\ref{subsec:secondary_predetermined}, instead of fixing $\vec{\beta}$, we can compute a $\vec{\beta}$ that results in the largest possible lower bound. Due to nonlinearity of the problem, however, we cannot optimize $\vec{\beta}$ and found maximum $\alpha$ in (\ref{eq:max_alpha_SC_controller}) simultaneously. The idea is to fix $\alpha$, compute the optimal $\vec{\beta}$ and $\eta$ using (\ref{eq:optimalBeta}), then update $\alpha$ using $\eta$ and repeat the process until $\alpha$ does not change by much. As in Section~\ref{subsec:secondary_predetermined}, we can use the $(\vec{\gamma},\vec{\beta})$-determined controller instead of the $\vec{\beta}$-determined controller to improve the obtained lower bound. The method is summarized in Module~\ref{module:gambeta_lowerbound}. When $\gamma=\beta$, Module~\ref{module:gambeta_lowerbound} provides a lower bound on $\alpha^{\max}$ like $\alpha^{(\beta)}$ based on $\vec{\beta}$-determined controllers.

Notice that $\lambda$ in Module~\ref{module:gambeta_lowerbound} should be set such that updates to $\alpha$ at each iteration are neither too large that the module falls into a loop, nor are too small that it takes a long time to converge.

\begin{proposition}\label{prop:beta_lowerbound}
When $\gamma=\beta$, for a good $\lambda$, Module~\ref{module:gambeta_lowerbound} converges to a $\alpha^{(\beta)}$ value such that $\alpha^{(\beta)}\leq \alpha^{\max}$. Moreover, $\alpha^*\leq \alpha^{(\beta)}$. (Recall that $\alpha^*$ is the optimal solution of (\ref{eq:max_alpha_SC_controller}).)
\end{proposition}

\begin{module}[t]
\caption{Lower Bound on $\alpha^{\max}$ using $(\vec{\gamma},\vec{\beta})$-determined Controllers}
\label{module:gambeta_lowerbound}
\small
\begin{trivlist}
\item\textbf{Input:} $G$, $\lambda$
\end{trivlist}
\begin{algorithmic}[1]
\STATE $\alpha^{(0)} = \hat{\alpha}$
\STATE flag = 1
\STATE i = 0
\WHILE{flag}
    \STATE flag = 0
    \STATE Compute the optimal value $\eta$, $\vec{\gamma}$, and $\vec{\beta}$ of (\ref{eq:optimalGamBeta}) for $\overline{p_d}=(1+\alpha^{(i)})\vec{p_d}^\dagger$ and $\underline{p_d}=(1-\alpha^{(i)})\vec{p_d}^\dagger$
    \STATE Set $\alpha^{(i+1)}=\alpha^{(i)}+\lambda (1-\eta)$
    \IF{$|\alpha^{(i+1)}-\alpha^{(i)}|>0.001$}
        \STATE flag = 1
        \STATE $i = i+1$
    \ENDIF
\ENDWHILE
\STATE \textbf{return} $\alpha^{(\gamma, \beta)}:=\alpha^{(i)}$, $\vec{\gamma}$, and $\vec{\beta}$
\end{algorithmic}
\end{module}

\begin{proposition}\label{prop:gambeta_lowerbound}
For a good $\lambda$, Module~\ref{module:gambeta_lowerbound} converges to a $\alpha^{(\gamma,\beta)}$ value such that $\alpha^{(\gamma,\beta)}\leq \alpha^{\max}$. Moreover, $\alpha^{(\beta)}\leq \alpha^{(\gamma,\beta)}$.
\end{proposition}

In the next section, we numerically compare the upper bound $\hat{\alpha}$, and lower bounds $\alpha^*$, $\alpha^{(\beta)}$, and $\alpha^{(\gamma,\beta)}$ with $\alpha^{\max}$ in order to demonstrate the tightness of these bounds in approximating $\alpha^{\max}$.

\section{Numerical Results}\label{sec:numerical}
In this section, we first numerically evaluate the performance of SAFE and IMMUNE Algorithms developed in Section~\ref{sec:primary}. 
Then, we numerically evaluate the accuracy of the upper and lower bounds developed in Section~\ref{sec:contingency} in approximating the maximum $\alpha$ such that the grid is $\alpha D$-robust (i.e., $\alpha^{\max}$).
\subsection{Simulations Setup}
For solving LP, we use \emph{CVX}, a package for specifying and solving convex programs~\cite{cvx,gb08}. For computing the optimal power flow part of the IMMUNE Algorithm, we use \emph{MATPOWER}~\cite{zimmerman2011matpower} which is a MATLAB based library for computing the power flows. We also exploit the power system test cases available with this library for our simulations. In particular, we use  IEEE 14-bus, 30-bus, and 57-bus test systems, and the New England 39-bus system.

The line capacities are only provided for the IEEE 30-bus and New England 39-bus   systems. Hence, for the other two systems, we set the capacities ourselves in two-different ways: (i) following \cite{hale2016ACDC} for each line we set $\overline{f_i}=\max\{1.2|f_i^\dagger|,\text{median}(|\vec{f}^\dagger|)\}$,  and (ii) set $\overline{f_i}=1.1\max(|\vec{f}^\dagger|)$, in which $\vec{f}^\dagger$ are the power flows given the default supply and demand values in the test systems. When the first method is used for determining the capacities, it is indicated by (f) in front of the grid name, and when the second method is used, it is indicated by (u) (e.g., see Table~\ref{tb:compare}).

\subsection{Primary Control}\label{subsec:prime_num}

\begin{table}[t]
\caption{Performance Evaluation of SAFE and IMMUNE Algorithms on New England 39-bus system. Cost values are in $\$/hr$. Numbers in parenthesis indicate the number of iterations took the IMMUNE Algorithm to converge.}
\vspace*{-0.2cm}
\centering
\begin{tabular}{|c|c|c|c|c|c|}
\hline
$\alpha$& OPF &SAFE&IMMUNE&IMMUNE-0.95&IMMUNE-0.9\\
\hline
0.09&41264&-&43434 (7)&43805 (4)&43859 (3)\\
\hline
0.08&41264&43628&42394 (8)& 42431 (3)&42982 (3)\\
\hline
0.07&41264&42665&41773 (5)& 41991 (3)&42405 (3)\\
\hline
0.06&41264&42050&41492 (4)& 41698 (3)&41534 (2)\\
\hline
0.05&41264&41668&41339 (10)& 41421 (3)&41419 (2)\\
\hline
\end{tabular}\label{tb:compare_primary}
\end{table}

\begin{figure}[t]
\centering
\includegraphics[scale=0.4]{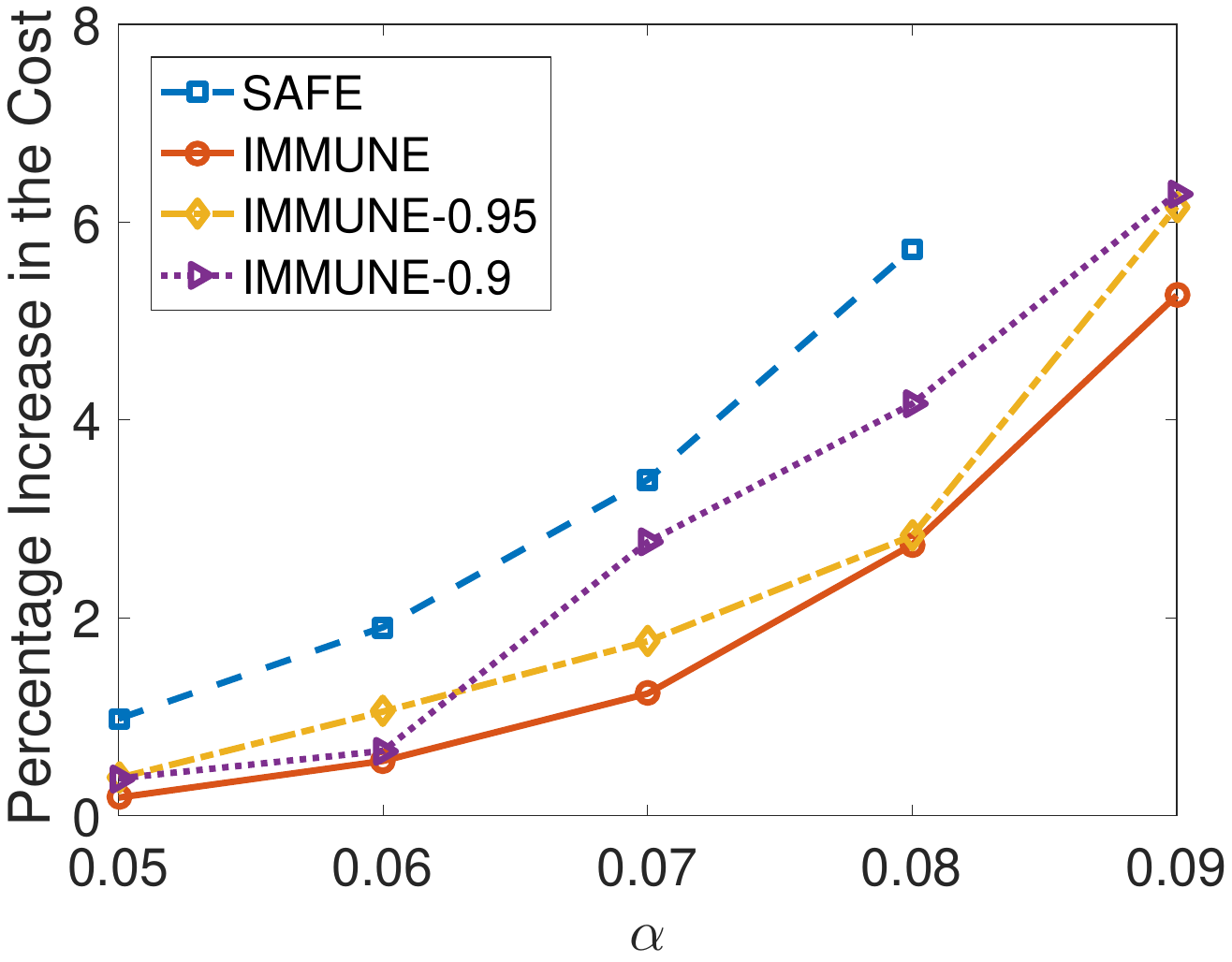}
\caption{Percentage increase in operating cost of the grid in order to make it robust against MAD attacks obtained by SAFE and IMMUNE Algorithms versus the magnitude of the attack ($\alpha$) in New England 39-bus system.}
\label{fig:primary}
\end{figure}

In this subsection, we evaluate the performance of SAFE and IMMUNE Algorithms on NEW England 39-bus and IEEE 30-bus systems. We assume that $(1-\alpha)p_{di}^\dagger\leq p_{di}\leq (1+\alpha) p_{di}^\dagger$ and consider different $\alpha$ values to capture attacks with different magnitudes (which depends on the number of controlled bots by an adversary).

Table~\ref{tb:compare_primary} compares the performance of SAFE and three variations of the IMMUNE Algorithm for different $\alpha$ values. Recall from Section~\ref{subsec:primary_nolimit} that IMMUNE-0.95 and IMMUNE-0.9 are similar to the IMMUNE Algorithm but apply more aggressive updates on the capacities in each iteration of the algorithm. This, as mentioned in Section~\ref{subsec:primary_nolimit} and demonstrated numerically here in Table~\ref{tb:compare_primary}, results in faster convergence of the algorithm. Since the OPF problem does not consider robustness of the grid against MAD attacks, its value is independent of the magnitude of an expected attack ($\alpha$).

As can be seen in Table~\ref{tb:compare_primary} and as we expected, the grid needs to be operated in a non-optimal operating point in order to be robust against MAD attacks. The required percentage increase in the operating cost of the grid obtained by the SAFE and IMMUNE Algorithms versus $\alpha$ are presented in Fig.~\ref{fig:primary}. IMMUNE Algorithm results in the least amount of increase in the operating cost. However, since as demonstrated in Table~\ref{tb:compare_primary}, IMMUNE Algorithm takes longer that IMMUNE-0.95 and IMMUNE-0.9 Algorithms to converge, the system operator may prefer to use IMMUNE-0.95 which performs approximately as well as the IMMUNE Algorithm but converges faster. Notice that due to nonconvexity of the problem, a more aggressive update rule may not necessary result in a costlier operating point, as we see here that IMMUNE-0.9 results in a lower operating cost than IMMUNE-0.95 for $\alpha=0.06$.

It can also be seen that SAFE Algorithm performs relatively well in finding a robust operating point of the grid much faster than all variations of IMMUNE Algorithm (recall from Section~\ref{subsec:primary_limit} that SAFE Algorithm requires only to solve a single LP). However, it may become infeasible for higher magnitude attacks (in this case for $\alpha=0.09$).

\begin{table}[t]
\caption{Performance Evaluation of SAFE and IMMUNE Algorithms on IEEE 30-bus system. Cost values are in $\$/hr$. Numbers in parenthesis indicate the number of iterations took the algorithm to converge.}
\vspace*{-0.2cm}
\centering
\begin{tabular}{|c|c|c|c|c|c|}
\hline
$\alpha$& OPF &SAFE&IMMUNE\\
\hline
0.31&565.2&-&- (3)\\
\hline
0.3&565.2&614.8&- (4)\\
\hline
0.28&565.2&571.6&569.6 (3)\\
\hline
0.26&565.2&565.32&565.22 (2)\\
\hline
0.22&565.2&565.2&565.2 (1)\\
\hline
\end{tabular}\label{tb:compare_primary30}
\end{table}

We repeated the simulations in IEEE 30-bus system. The results are presented in Table~\ref{tb:compare_primary30}. First, it can be seen that the IEEE 30-bus system can be protected against much stronger attacks ($\alpha=0.3$) which demonstrates that different grids may have different levels of robustness against MAD attacks (we will make a similar observation in the secondary control case in the next subsection). Unlike the New England 39-bus case, here the IMMUNE Algorithm does not converge for the strongest attack ($\alpha=0.3$) rather than the SAFE Algorithm. This demonstrates that each of these algorithms may be useful in finding a robust operating point for the grid in different scenarios--besides their running time and optimality.

As can be seen in Table~\ref{tb:compare_primary30}, in this case also, if the IMMUNE Algorithm converges, it converges to a lower cost operating point than the one obtained by the SAFE Algorithm. Here, the IMMUNE Algorithm converged within few iterations. Therefore, there were no need to consider faster variation of the IMMUNE Algorithm as in the New England 39-bus case.

Finally, it can be seen that for $\alpha=0.31$, none of the algorithms can obtain a robust operating point for the grid. We show in the next subsection that this case can be handled by the secondary controller instead (assuming that lines can handle temporary overloads).

\subsection{Secondary Control}\label{subsec:sec_num}
In order to evaluate the performance of the controllers developed in Section~\ref{subsec:secondary_predetermined}, in this subsection, we focus on their performance in approximating $\alpha^{\max}$ as described in Section~\ref{sec:contingency}.

Table~\ref{tb:compare} compares the maximum $\alpha$ obtained by different methods in several test cases. As can be seen and proved in Section~\ref{sec:contingency}, in all cases, $\alpha^*\leq\alpha^{(\beta)}\leq\alpha^{(\gamma,\beta)}\leq\alpha^{\max}\leq\hat{\alpha}$. Notice that for the IEEE 57-bus system, since the brute force search algorithm needs to solve (\ref{eq:minPg}) about $2^{42}$ times for each given $\alpha$ to determine the secondary controllability of the grid, we could not exactly determine $\alpha^{\max}$. However, in the case of IEEE 57-bus (f), after initial iterations of the brute force search algorithm, we could determine that the grid is not secondary controllable for $0.09\leq \alpha$ as presented in the table.

It can be seen that $\hat{\alpha}$ provides a very close upper bound for $\alpha^{\max}$ most of the time (except in IEEE 57-bus (f)). And since it can be computed by a single LP, the numerical results suggest that it is an efficient and reliable way to find an upper bound for $\alpha^{\max}$. On the other hand, $\alpha^*$ that can also be computed efficiently by a single LP, does not provide a very close lower bound in the test systems that we studied here. However, $\alpha^{(\beta)}$ and $\alpha^{(\gamma,\beta)}$ that require more time to be computed, provide much better lower bounds. In particular, in the case of New England 39-bus system $\alpha^{(\gamma,\beta)}=\hat{\alpha}$ which implies that $\alpha^{\max}=\alpha^{(\gamma,\beta)}=\hat{\alpha}$.

\begin{table}[t]
\caption{Lower and upper bounds for $\alpha^{\max}$.}
\vspace*{-0.2cm}
\centering
\begin{tabular}{|l|c|c|c||c||c|}
\hline
Test case& $\alpha^*$ &$\alpha^{(\beta)}$&$\alpha^{(\gamma,\beta)}$&$\alpha^{\max}$&$\hat{\alpha}$\\
\hline
{\ttfamily \small IEEE 14-bus (f)}& 0.058 & 0.1649 & 0.1906&0.2117&0.2117\\
\hline
{\ttfamily \small IEEE 14-bus (u)}& 0.950 & 1.0243 & 1.1454&1.1479&1.1479\\
\hline
{\ttfamily \small IEEE 30-bus}& 0.214 & 0.2851 & 0.3126&0.37&0.3717\\
\hline
{\ttfamily \small NE 39-bus}& 0.039 & 0.0796 &0.0962&0.0962&0.0962\\
\hline
{\ttfamily \small IEEE 57-bus (f)} & 0.024& 0.0307& 0.0311&$<0.09$&0.2\\
\hline
{\ttfamily \small IEEE 57-bus (u)}& 0.128 & 0.2396 & 0.2864&-&0.3468\\
\hline
\end{tabular}\label{tb:compare}
\end{table}

\begin{figure}[t]
\centering
\includegraphics[scale=0.35]{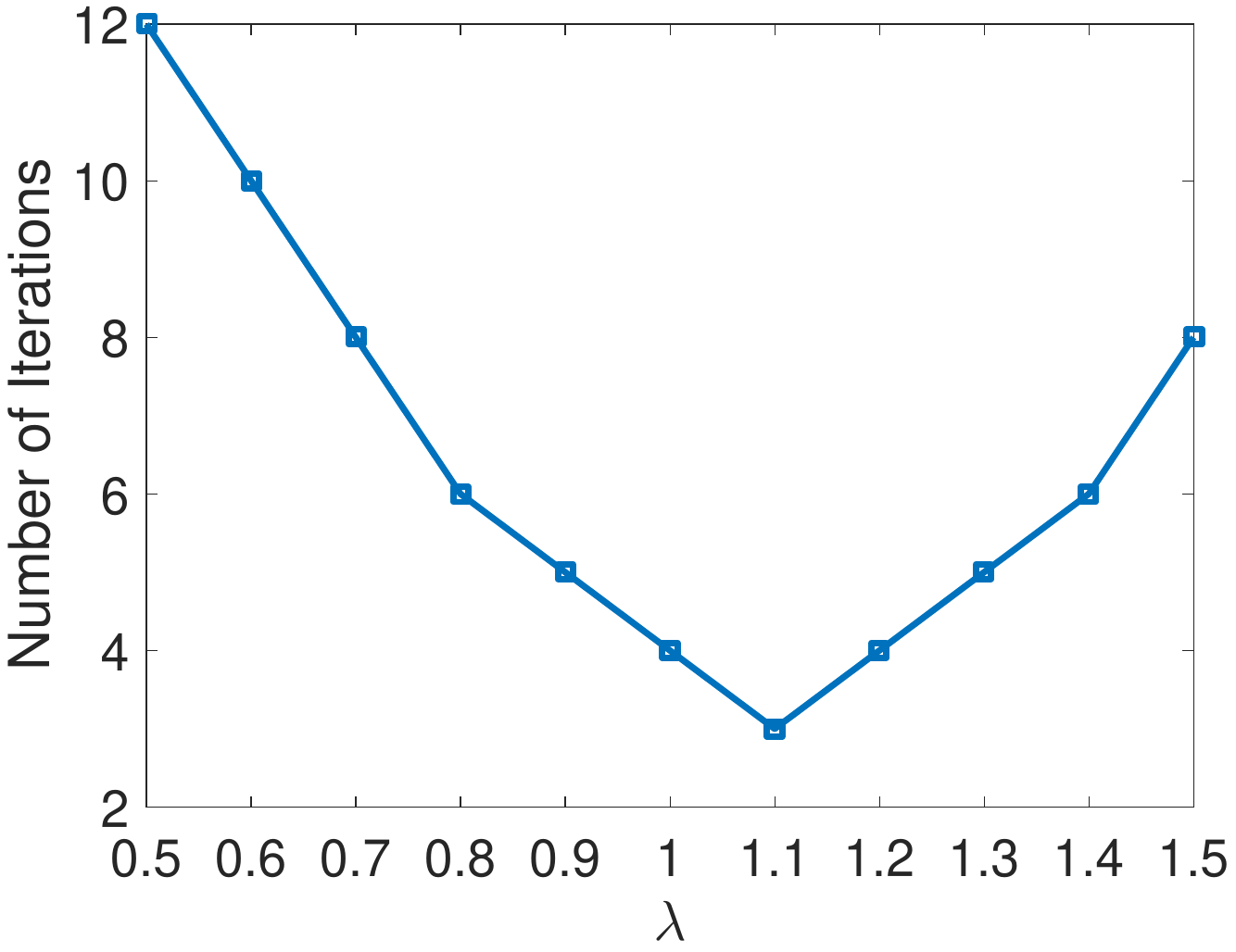}
\caption{Number of iterations in Module~\ref{module:gambeta_lowerbound} before it converges versus its update step size $\lambda$ in IEEE 30-bus system.}
\label{fig:runVSlambda}
\end{figure}

Although finding $\alpha^{(\beta)}$ and $\alpha^{(\gamma,\beta)}$ requires solving an LP in several iterations (as summarized in Module~\ref{module:gambeta_lowerbound}), the number of iterations can be minimized by selecting a good step size $\lambda$. For example, the number of iterations of Module~\ref{module:gambeta_lowerbound} versus $\lambda$ is presented in Fig.~\ref{fig:runVSlambda} in the IEEE 30-bus system. As can be seen, for the optimal $\lambda$ (in this case $\lambda=1.1$), the module converges in 3 iterations. Hence, it can find a good lower bound for $\alpha$, as shown in Table~\ref{tb:compare}, very efficiently and in polynomial time (since it solves a single LP at each iteration). A good $\lambda$ can be found in practice heuristically after the first few iterations and observing the rate of changes.

Finally, as mentioned in Section~\ref{sec:secondary}, the secondary controllability becomes more important when the primary controller cannot prevent line overloads, but the overloads can be tolerated for a short period of time. An example of such scenario happens in IEEE 30-bus system and when $\alpha=0.31$. As can be seen in Table~\ref{tb:compare_primary30}, none of the SAFE and IMMUNE Algorithms can find a robust operating point for the grid in this case. However, as can be seen in Table~\ref{tb:compare}, since this value is less that $\alpha^{\max}=0.37$, any line overloads can be cleared by the secondary controller.

\section{Conclusions}\label{sec:conclusions}
In this paper, we analyzed the effect of MAD attacks on power flows in detail and presented SAFE and IMMUNE algorithms for finding robust operating points for the generators during economic dispatch such that no lines are overloaded after automatic primary control response to any MAD attacks. Moreover, we demonstrated that in cases that temporary overloads can be tolerated, the system operator can approximately but efficiently verify in advance that potential line overloads can be cleared during the secondary control after any MAD attacks. Based on that, we defined $\alpha D$-robustness notion and demonstrated that upper and lower bounds on the maximum $\alpha$ for which the grid is $\alpha D$-robust can be found efficiently and in polynomial time. We finally evaluated the performance of the developed algorithms and methods, and showed that they perform very well in practical test cases.

We believe that with universality and growth in the number of high-wattage IoT devices and smart thermostats, the probability of MAD attacks is increasing and there is an urgent need for more studies on the potential effects of these attacks and developing tools for grid protection. Our work provides the first methods for protecting the grid against potential line failures caused by newly discovered MAD attacks. However, our work can be extended in several directions. A straight forward direction is to extend the developed results to the AC power flow model. A more challenging research direction is to extend the methods to unit commitment phase of the grid operation. Since regular unit commitment problem is already a combinatorial problem, incorporating security constraints into that problem will be a challenging task and part of our future work.



\bibliographystyle{ACM-Reference-Format}
\bibliography{Ant_bib}
\appendix
\section{Omitted Proofs}
\begin{proof}[Proof of Lemma~\ref{lem:gen_order}]
First, notice that $1/R_i$ is the rate with which generator $i$ increases its generation to compensate for the extra demand. Hence, $t_i$ denotes the time that generator $i$ reaches its maximum capacity if the total supply does not meet the demand before $t_i$. Accordingly, generators reach their maximum capacity in order of their $t_i$ values from smallest to largest.  Using this, it is easy to see that $S_i$ is the total change in generation at time $t_i$. Therefore, if $S_i<S_{\Delta p_d}$, then generators $1$ to $i$ will reach their maximum capacities before supply meets the total demand. Moreover, since $S_{\Delta p_d}\leq S_{i+1}$, generators $i+1,\dots,n$ do not reach their capacities and each contribute according to their droop characteristic to compensate for the remaining $S_{\Delta p_d}-\sum_{l=1}^i (\overline{p_{gl}}-p_{gl})$.
\end{proof}
\begin{proof}[Proof of Lemma~\ref{lem:IMMUNE_converge}]
First, notice that for each line $(i,j)\in E$ and in each iteration of the IMMUNE Algorithm, $c_{ij}$ is not increasing. To see this, assume $c_{ij}$ changes in the $l^{th}$ iteration, and $c_{ij}^{\text{old}}$ and $c_{ij}^{\text{new}}$ denote the value of $c_{ij}$ before and after the change, respectively. Since $c_{ij}$ is changed, it means that $\overline{f_{ij}}<|f_{ij}|+\Delta f_{ij}^{\max}$. On the other hand, $|f_{ij}|\leq c_{ij}^{\text{old}}$. Hence, $\overline{f_{ij}}<c_{ij}^{\text{old}}+\Delta f_{ij}^{\max}$ or $\overline{f_{ij}}-\Delta f_{ij}^{\max}<c_{ij}^{\text{old}}.$ Since $c_{ij}^{\text{new}}=\overline{f_{ij}}-\Delta f_{ij}^{\max}$, therefore $c_{ij}^{\text{new}}<c_{ij}^{\text{old}}$.

On the other hand, from (\ref{eq:upperbound1}), it is easy to verify that after each iteration $\overline{f_{ij}}-\widehat{\Delta f_{ij}}\leq c_{ij}$. Hence, $c_{ij}$s cannot get smaller than the fixed values $\overline{f_{ij}}-\widehat{\Delta f_{ij}}$ and since (\ref{eq:ec_dis2}) is feasible, the OPF problem remains feasible after each iteration of the IMMUNE algorithm. Now since $c_{ij}$s are non-increasing and limited by lower bounds, the algorithm is guaranteed to remain feasible and converge to a local optimum solution.
\end{proof}
\begin{proof}[Proof of Lemma~\ref{lem:immune_runtime}]
In each iteration of the IMMUNE algorithm, at least for a single line $(i,j)$, the $c_{ij}$ will be updated. Otherwise the algorithm should terminate (either converges or become infeasible). On the other hand, since $\widehat{\Delta f_{ij}}$ is the maximum possible flow change on line $(i,j)$, the $c_{ij}$ cannot get smaller than $\overline{f_{ij}}-\widehat{\Delta f_{ij}}$. Hence, since the updates are discrete, in the worst case that only a single capacity is updated by a single unit at each iteration, the algorithm can take at most $\sum_{(i,j)\in E} \lceil \widehat{\Delta f_{ij}}\rceil$ iterations to terminate.
\end{proof}
\begin{proof}[Proof of Lemma~\ref{lem:extreme}]
Assume $\vec{p_d}^{(1)},\vec{p_d}^{(2)},\dots, \vec{p_d}^{(2^n)}$ denote all possible extreme demand vectors. Now assume that for each extreme demand vector $\vec{p_d}^{(i)}$, there exists a operating vector $\vec{p_g}^{(i)}$ for generators that satisfies the secondary control conditions. We prove that for all demand vectors $\vec{p_d}$ within the upper and lower limits also there exists an operating vector $\vec{p_g}$ that satisfies all the secondary controller conditions. Since the space of all the demand vectors is convex, each demand vector $\vec{p_d}$ within the upper and lower limits can be written as a convex combination of the extreme points such as $\vec{p_d} = \sum_{i=1}^{2^n} \beta_i \vec{p_d}^{(i)}$ in which $\forall i: \beta_i\geq0$ and $\sum_{i=1}^{2^n}\beta_i=1$. We show that $\vec{p_g}= \sum_{i=1}^{2^n} \beta_i \vec{p_g}^{(i)}$ satisfies all the secondary controller conditions. First, since $\vec{p_g}$ is a convex combination of $\vec{p_g}^{(i)}$s and they are within generators upper and lower limits, so is $\vec{p_g}$. Second, it is easy to see that $\vec{1}^T(\vec{p_g}-\vec{p_d})=\sum_{i=1}^{2^n} \beta_i \vec{1}^T(\vec{p_g}^{(i)}-\vec{p_d}^{(i)})=\sum_{i=1}^{2^n} \beta_i 0 = 0$. Finally, based on our assumptions, for each $i$: $-\overline{f}\leq \B(\vec{p_g}^{(i)}-\vec{p_d}^{(i)})\leq \overline{f}$. Hence, $\B(\vec{p_g}-\vec{p_d})=\sum_{i=1}^{2^n} \beta_i \B(\vec{p_g}^{(i)}-\vec{p_d}^{(i)})\leq \sum_{i=1}^{2^n} \beta_i \overline{f}=\overline{f}$. Similarly, $-\overline{f}\leq \B(\vec{p_g}-\vec{p_d})$. Therefore, $\vec{p_g}$ satisfies all the constraints of the secondary controller problem. The reverse can also be similarly proved using contradiction method.
\end{proof}
\begin{proof}[Proof of Proposition~\ref{prop:second_approx_beta}]
If there exists a vector $\vec{\beta}$ that the $\vec{\beta}$-determined controller is reliable, then for any feasible demand vector $\vec{p_d}$, vector of operating points $\vec{p_g}=(\sum_{i=1}^n p_{di})\times \vec{\beta}$ satisfies the demands (i.e., $\vec{1}^T(\vec{p_g}-\vec{p_d})=0$) and $|\vec{f}|=|\B(\vec{p_g}-\vec{p_d})|\leq \overline{f}$. Therefore, the grid is secondary controllable.
\end{proof}
\begin{proof}[Proof of Lemma~\ref{lem:maxflow_beta}]
From definition of $\vec{w_i}^{(\beta)}$ vectors, it is easy to verify that for a demand vector $\vec{p_d}$, the power flow on line $e_k$ can be computed as $f_k=\sum_{i=1}^n p_{di}\B_k \vec{w_i}^{(\beta)}$. For $|f_k|$ to be maximized, each $p_{id}$ should be either equal to $\underline{p_{di}}$ or $\overline{p_{di}}$ based on signs of  $\B_k \vec{w_i}^{(\beta)}$ and $f_k$. On the other hand, it is easy to see that $\underline{p_{di}}=\frac{(\overline{p_{di}}+\underline{p_{di}})}{2}-\frac{(\overline{p_{di}}-\underline{p_{di}})}{2}$ and $\overline{p_{di}}=\frac{(\overline{p_{di}}+\underline{p_{di}})}{2}+\frac{(\overline{p_{di}}-\underline{p_{di}})}{2}$. So by considering only $p_{di}\in\{\overline{p_{di}},\underline{p_{di}}\}$, $f_k$ can be computed as follows:
\begin{align*}
f_k &= \sum_{i=1}^n p_{di}\B_k \vec{w_i}^{(\beta)} = \sum_{i=1}^n \big(\frac{(\overline{p_{di}}+\underline{p_{di}})}{2}\pm\frac{(\overline{p_{di}}-\underline{p_{di}})}{2}\big)\B_k \vec{w_i}^{(\beta)}\\
&= \sum_{i=1}^n \frac{(\overline{p_{di}}+\underline{p_{di}})}{2}\B_k \vec{w_i}^{(\beta)}+\sum_{i=1}^n\big(\pm\frac{(\overline{p_{di}}-\underline{p_{di}})}{2}\big)\B_k \vec{w_i}^{(\beta)}.
\end{align*}
From equation above, it can be seen that the first part is fixed but the second part can be selected based on sign of the first part in order to maximize $|f_k|$. Hence, it is easy to see that maximum value of $|f_k|$ is:
\begin{align*}
\max_{\underline{p_d}\leq \vec{p_d}\leq \overline{p_d}} |f_k| = \left|\sum_{i=1}^n \frac{(\overline{p_{di}}+\underline{p_{di}})}{2} \B_k\vec{w_i}^{(\gamma)}\right| + \sum_{i=1}^n \frac{(\overline{p_{di}}-\underline{p_{di}})}{2}|\B_k\vec{w_i}^{(\beta)}|.
\end{align*}
\end{proof}
\begin{proof}[Proof of Proposition~\ref{prop:LPBeta}]
In order to solve (\ref{eq:optimalBeta}) using LP, one can define auxiliary vector $\vec{u}$ and matrix $\Q$ and replace the constraint $\vec{f} = |\B\W^{(\beta)}(\overline{p_d}+\underline{p_d})/2| + |\B\W^{(\beta)}|(\overline{p_d}-\underline{p_d})/2$ in (\ref{eq:optimalBeta}) with following set of inequalities:
\begin{align*}
& \vec{f} = \vec{u} + \Q(\overline{p_d}-\underline{p_d})/2,\\
& \vec{u}\geq \B\W^{(\beta)}(\overline{p_d}+\underline{p_d})/2,\\
& \vec{u}\geq -\B\W^{(\beta)}(\overline{p_d}+\underline{p_d})/2,\\
& \Q \geq \B\W^{(\beta)},~~\Q \geq -\B\W^{(\beta)},
\end{align*}
in which the matrix inequalities are entry by entry. Now it is easy to verify that since the optimization minimize $\eta$ and $\vec{f}\leq \eta \overline{f}$, in the optimal solution $\vec{f}$ will be minimized and therefore $\vec{u}$ and $\Q$ will be equal to $|\B\W^{(\beta)}(\overline{p_d}+\underline{p_d})/2|$ and $|\B\W^{(\beta)}|$, respectively. Hence using above transformation, (\ref{eq:optimalBeta}) can be solved using LP. It can be seen that if the optimal solution $\eta^*$ to (\ref{eq:optimalBeta}) is less than or equal to 1, then since $\vec{f}$ is equal to the maximum power flow on the lines over all possible demand vectors (and corresponding generation operating points obtained by the $\vec{\beta}^*$-determined controller) and $\vec{f}\leq \eta^*\overline{f}\leq \overline{f}$, the $\vec{\beta}^*$-controller is reliable. Hence, the grid is secondary controllable.
\end{proof}
\begin{proof}[Proof of Proposition~\ref{prop:alpha_hat}]
Since in optimization (\ref{eq:max_alpha}) only the maximum demand case (i.e., $\vec{p_d}=(1+\alpha)\vec{p_d}^\dagger$) is being verified to be satisfiable by the generators with no line overloads, the optimal solution of (\ref{eq:max_alpha}) only provides an upper bound for $\alpha^{\max}$.
\end{proof}
\begin{proof}[Proof of Proposition~\ref{prop:alpha_star}]
Using (\ref{eq:maxflow_beta}), it can be verified that the maximum power flow on a line $(i,j)$ over all the demand vectors and corresponding generation vector determined by the $\vec{\beta}$-determined controller is equal to $|\B\W^{(\beta)}\vec{p_d}^\dagger| + |\B\W^{(\beta)}|(\alpha\vec{p_d}^\dagger)$. Hence, optimization (\ref{eq:max_alpha_SC_controller}) maximizes $\alpha$ such that the grid is $\alpha D$-robust using the specified $\vec{\beta}$-determined controller. On the other hand, since the operating points of the generators are limited to the operating points obtained by the specified $\vec{\beta}$-determined controller, it is obvious that demand vectors that are controllable by this controller are subset of all controllable vectors. Hence, $\alpha^*$ only provides a lower bound for $\alpha^{\max}$. Finally, it is also easy to see that similar to the technique presented in the proof of Proposition~\ref{prop:LPBeta}, optimization (\ref{eq:max_alpha_SC_controller}) can be solved using LP and therefore $\alpha^*$ can be computed in polynomial time.
\end{proof}
\begin{proof}[Proof of Proposition~\ref{prop:beta_lowerbound}]
At each iteration, if $\alpha^{(i)}>\alpha^{\max}$, then the solution $\eta$ to (\ref{eq:optimalBeta}) would be greater than 1. Hence, if $\lambda$ is small enough, $0\leq \alpha^{(i+1)}= \alpha^{(i)}+\lambda (1-\eta)\leq \alpha^{(i)}$. Similarly, it can be shown that if $\alpha^{(i)}<\alpha^{\max}$,  then $\alpha^{(i+1)}>\alpha^{(i)}$. On the other hand, for $\alpha^{(i)}=\alpha^{\max}$, the solution $\eta$ to (\ref{eq:optimalBeta}) would be zero and $\alpha^{(i)}=\alpha^{(i+1)}=\alpha^{\max}$. Hence, $\alpha^{\max}$ is the only absorbing point for this algorithm which it converges to (if $\lambda$ is small enough).
\end{proof}
\begin{proof}[Proof of Proposition~\ref{prop:gambeta_lowerbound}]
The convergence proof is similar to the proof of Proposition~\ref{prop:beta_lowerbound}. It is also easy to see that since $\vec{\beta}$-determined controllers are  special case of $(\vec{\gamma},\vec{\beta})$-determined controllers, $\alpha^{(\beta)}\leq \alpha^{(\gamma,\beta)}$.
\end{proof}
\end{document}